\documentclass[12pt]{article}
\usepackage[natbibapa]{apacite}   %

\usepackage[a4paper,left=1.1in,right=1.1in,top=1in,bottom=1.2in,footskip=.7in]{geometry}
\usepackage{mathtools}
\mathtoolsset{showonlyrefs=true}
\usepackage{amssymb,amsthm}
\usepackage{newtxtext}
\usepackage[subscriptcorrection]{newtxmath}
\usepackage{graphicx}
\usepackage{booktabs}
\usepackage{url}
\usepackage[hypertexnames=false]{hyperref}
\hypersetup{colorlinks=true,allcolors=black}  %
\usepackage{xcolor}
\newif\ifrevmark
\revmarkfalse

\usepackage{caption}
\makeatletter
\AtBeginDocument{%
  \@ifpackageloaded{apacite}{%
    \renewcommand{\doi}[1]{\url{https://doi.org/#1}}%
  }{}%
}
\makeatother

\graphicspath{{Figures/}}

\usepackage{setspace}
\newcommand{\cY}{\mathcal{Y}}

\newcommand{\br}{\mathbf{r}}
\newcommand{\bq}{\mathbf{q}}

\newcommand{\balpha}{{\boldsymbol{\alpha}}}

\newcommand{\papertitle}{An Exact Marginal Likelihood for Inference under Careless Responding}

\providecommand{\bp}{\mathbf{p}}

\newenvironment{keywords}{\par\medskip\noindent\textbf{Keywords} }{\par\medskip}

\newtheorem{proposition}{Proposition}
\newtheorem{theorem}{Theorem}
\newtheorem{corollary}{Corollary}
\newtheorem{lemma}{Lemma}[section]
\theoremstyle{remark}
\newtheorem{remark}{Remark}

\title{\papertitle}
\author{Jaehoan Kim\\
Department of Statistical Science, Duke University\\
Durham, North Carolina 27708, U.S.A.\\
\texttt{jaehoan.kim@duke.edu}}
\date{November 6, 2026}

\begin{document}
\maketitle

\begin{abstract}
Careless or insufficient-effort responding threatens the validity of self-report measurement. Mixture models for careless responding specify its form, and screening rules and robust methods depend on a threshold or tuning constant with no direct meaning in terms of careless responses. We propose ExHuber, a Bayesian method for inference on the engaged population under Huber's contamination model, in which an unknown proportion of responses comes from a contaminating distribution of unspecified form instead of the structural model. With a finite support of responses, such as the values of a scale score, a Dirichlet prior and a Beta prior on the contaminating distribution and its proportion yield a closed-form exact marginal likelihood of the structural parameters. Under a prior favoring a contaminating distribution on a few values, we show that the posterior of the structural parameters is consistent when careless responses occupy few values of the support, without knowing which values. Dynamic programming evaluates the marginal likelihood and its gradient at a cost quadratic in the sample size, which enables gradient-based sampling. Simulations from graded response models show accurate inference under straightlining and repeated response patterns, and an application to Big Five scales estimates the engaged mean of each scale.
\end{abstract}

\begin{keywords}
Careless responding, Huber contamination, Marginal likelihood, Bayesian robustness.
\end{keywords}

\section{Introduction}

Survey data, such as responses to self-report scales, are widely used in psychological research \citep{WardMeade2023}. A fraction of respondents, however, respond carelessly, providing answers that do not reflect the item content or their true trait levels \citep{MeadeCraig2012,Curran2016,HuangEtAl2012,WardMeade2023}. Such responses threaten the validity of conclusions drawn from these data, as they can reduce reliability, shift scale means, attenuate or inflate correlations and distort factor structure \citep{AriasEtAl2020,KamCheung2024,WardMeade2023, WelzEtAl2026impact}. This motivates inference on the engaged population under an unknown fraction and form of careless responses, without knowing which respondents answered carelessly.

In the measurement literature, careless responding is addressed by indicators and screening rules and by mixture models. Indicators, attention checks and unsupervised learning flag individual respondents, often through thresholds chosen by the analyst \citep{OppenheimerMeyvisDavidenko2009,Johnson2005,Curran2016,AlfonsWelz2024}. Common indicators include longstring indices, consistency indicators such as psychometric synonyms, multivariate outlier analyses and response time thresholds \citep{Johnson2005,MeadeCraig2012,Curran2016,AlfonsWelz2024}. Mixture models instead place careless responses in one or more specified classes \citep{KamCheung2024,UglanovaEtAl2026}. These models use raw item responses and can be fit in standard software \citep{Chalmers2012,KamCheung2024}. Factor mixture models separate respondents who answer consistently across positively and negatively worded items from those who do not \citep{AriasEtAl2020,KamCheung2024}, and mixture item response models add a class governed by category preferences or random responding \citep{vanLaarBraeken2024,UglanovaEtAl2026}, with Bayesian estimation also available \citep{ZhangEtAl2025a,ZhangEtAl2025b}. Attentiveness can also vary within respondents, through response times \citep{UlitzschEtAl2022}, item characteristics \citep{UlitzschEtAl2022b}, or an estimated onset of carelessness \citep{WelzAlfons2024}. These approaches specify the attentive measurement model and the form of careless responding. Robust estimators for categorical data and polychoric correlations instead downweight such responses \citep{RuckstuhlWelsh2001,Welz2024,WelzEtAl2026}.

In the statistics literature, this problem is studied in robust statistics, which concerns estimation or inference that is not affected by a small fraction of contaminated or misspecified data \citep{Huber1964,HuberRonchetti2009}. The standard formalization is Huber's contamination model,
\begin{equation}\label{eqn: Huber contamination}
P_0=(1-\epsilon)\,P_\theta+\epsilon\,Q,
\end{equation}
where $P_\theta$ is the structural model for engaged responses, $\epsilon\in[0,1]$ is the unknown proportion of contamination, and $Q$ is a contaminating distribution with no parametric form. %

Under such a model, Bayesian approaches replace the likelihood by a loss or divergence \citep{BissiriHolmesWalker2016,GhoshBasu2016} or condition on a neighborhood of the empirical distribution \citep{MillerDunson2019}. Mixture approaches model the contamination directly with parametric or nonparametric structure on $Q$ \citep{VerdinelliWasserman1991,ShotwellSlate2011,PageDunson2011,CamerlenghiCorradinOngaro2022}.

Both lines of work face two challenges. First, mixture models specify the form of careless responding or of $Q$. Careless responding takes several forms \citep{WardMeade2023}. The dominant form in a given survey is typically unknown, and a misspecified form or dimensionality can bias the estimated careless proportion and the resulting conclusions \citep{UglanovaEtAl2026,VogelsmeierEtAl2024}.  Second, robust estimators and generalized posteriors depend on a tuning constant, such as a weight or a power, that has no direct meaning in terms of the amount or form of careless responding, and how to set it so that the posterior is calibrated remains open \citep{BissiriHolmesWalker2016,MillerDunson2019}.

We propose an exact marginal likelihood for Huber's contamination model (ExHuber), which addresses these challenges. With a prespecified support including all the measurement values, the contaminating distribution $Q$ can be considered as a probability vector. Then a Dirichlet prior on $Q$ and a Beta prior on $\epsilon$ yield a closed-form exact marginal likelihood of $\theta$. This construction can be applied to any structural model with a finite support, including the latent correlation of two ordinal items. Since the support of the full item response pattern grows exponentially in the number of items, we apply it to a low-dimensional summary such as a scale score. Under the ExHuber model, the hyperparameters of the priors on $Q$ and $\epsilon$ have a direct meaning. The Beta prior on $\epsilon$ can encode domain knowledge on careless proportions, such as proportions reported in comparable surveys \citep{WardMeade2023,ArthurEtAl2021}, and the Dirichlet concentration encodes how many values careless responses are expected to occupy. We prove that under a Dirichlet prior favoring $Q$ on a few values, the posterior of $\theta$ is consistent when careless responses occupy only a few values of the support, without knowing which values. This is the case for straightlining and repeated response patterns. Under a uniform Dirichlet prior, which is suitable for settings with little knowledge on $Q$, the marginal likelihood favors the $\theta$ requiring the smallest contamination fraction to represent $P_0$, and it favors the true $\theta$ when there is no contamination. We then introduce a dynamic programming approach to calculate the marginal likelihood and its gradient under ExHuber in $O(n^2)$ operations for a fixed support, which enables gradient-based posterior sampling. We investigate ExHuber through simulations from graded response models and an application to five Big Five scales, with the engaged mean of a scale score as the target. %

The remainder of the paper is organized as follows.
Section~\ref{sec: model} derives the finite-support Huber marginal likelihood.
Section~\ref{sec: theory} gives its large-sample interpretation.
Section~\ref{sec: computation} presents the dynamic programming algorithm and score recursion.
Sections~\ref{sec: simulation} and~\ref{sec: real data} evaluate the method on synthetic and real data, and Section~\ref{sec: discussion} concludes the paper.

\section{Huber contamination model for careless responding}
\label{sec: model}

We consider responses that take finitely many values, and we denote their support by $\cY=\{z_1,\ldots,z_m\}$. As in \eqref{eqn: Huber contamination}, $P_\theta$ denotes a structural model for an engaged respondent. The support $\cY$ and $P_{\theta}$ depend on the measurement situation. For instance, for a Likert sum score, $z_1,\ldots,z_m$ are the possible scores, and one possible choice of $P_{\theta}$ is the beta-binomial distribution \citep{KeatsLord1962}. For two ordinal items, $z_1,\ldots,z_m$ are the cells of the contingency table, and $P_\theta$ is the table implied by the thresholds and a latent correlation. For a single categorical or ordinal item or a symptom count, they are the response categories or the counts. Modeling the full joint pattern of item responses is impractical because its support grows exponentially in the number of items. Therefore, we consider a low-dimensional summary such as a scale score. %

\subsection{Exact marginal likelihood}
\label{sec: collapsed likelihood}

Consider $n$ observations $Y_1,\ldots,Y_n$ drawn from a population law $P_0$ identically and independently on $\cY$.
The structural model is $\{P_\theta:\theta\in\Theta\}$ with atom probabilities
\begin{equation}
p_j(\theta)=P_\theta(Y=z_j),
\qquad j=1,\ldots,m.
\label{eqn: structural atom probabilities}
\end{equation}
We write $\bp(\theta)=(p_1(\theta),\ldots,p_m(\theta))$.
Huber's contamination model represents the population law as
\begin{equation}
P_0=(1-\epsilon)P_\theta+\epsilon Q,
\label{eqn: huber model}
\end{equation}
where $\epsilon\in[0,1]$ is the contamination proportion and $Q$ is a single unrestricted distribution on $\cY$, with $q_j=Q(\{z_j\})$ and $p_{0j}=P_0(Y=z_j)$. Both the proportion $\epsilon$ and the contaminating distribution $Q$ are unknown, and the probability vector $\bq = (q_1, \ldots, q_m)$ specifies the contaminating distribution $Q$ under the finite support $\cY$.
The contaminating distribution $Q$ encodes the departures from $P_\theta$ on the recorded support, including careless responding and any lack of fit of $P_\theta$.

The parameter of interest is $\theta$, the structural parameter of the engaged-population law $P_\theta$ read through the chosen finite support $\cY$. It is a population-level quantity, and a different finite-support summary would in general define a different $\theta$. The contaminating distribution $Q$ and its proportion $\epsilon$ are nuisance components.

Conditional on $(\theta,\epsilon,\bq)$, the data are independent draws from the atom probabilities $(1-\epsilon)p_j(\theta)+\epsilon q_j$. Writing $N_j=\sum_{i=1}^n\mathbf{1}(Y_i=z_j)$ for the observed count at atom~$z_j$, the conditional likelihood is
\begin{equation}
L(\theta,\epsilon,\bq;y_{1:n})
:=\Pr(Y_1=y_1,\ldots,Y_n=y_n\mid\theta,\epsilon,\bq)
=
\prod_{j=1}^m\{(1-\epsilon)p_j(\theta)+\epsilon q_j\}^{N_j},
\label{eqn: finite support likelihood}
\end{equation}
with $\sum_{j=1}^m N_j=n$.

\begin{remark}\label{rem: non-iid}
The construction extends to independent observations that are not identically distributed. Suppose observation $i$ has structural atom probabilities $p_{ij}(\theta)=P_\theta(Y_i=z_j\mid x_i)$, for example through a covariate $x_i$, while the contaminating distribution $Q$ and the proportion $\epsilon$ are shared across observations. The conditional likelihood becomes $\prod_{i=1}^n\prod_{j=1}^m\{(1-\epsilon)p_{ij}(\theta)+\epsilon q_j\}^{\mathbf{1}(Y_i=z_j)}$ and the same priors on $\bq$ and $\epsilon$ still integrate out in closed form when $Q$ is shared. This covers regression models in which $\theta$ governs how the atom probabilities depend on covariates. 
\end{remark}

Next, we assign independent priors, also independent of $\theta$, on the nuisance parameters $\bq$ and $\epsilon$ as
\begin{equation}
\bq\sim\operatorname{Dirichlet}(\alpha_1,\ldots,\alpha_m),
\qquad
\epsilon\sim\operatorname{Beta}(a,b),
\label{eqn: nuisance priors}
\end{equation}
where $\alpha_j>0$, $a>0$, $b>0$, and $\alpha_0=\sum_{j=1}^m\alpha_j$. It is difficult to identify $\epsilon$ and $Q$ from the data \citep{Manski2003,MossGronneberg2023}, so we model $\bq$ and $\epsilon$ as latent parameters with priors. 
We can integrate them out of \eqref{eqn: finite support likelihood}, which gives the marginal likelihood of $\theta$
\begin{equation}
\bar L(\theta)
:= \Pr(Y_1 = y_1, \ldots, Y_n = y_n \mid \theta) = 
\int_0^1\int_{\Delta_{m-1}}
L(\theta,\epsilon,\bq;y_{1:n})\,\pi(\bq)\,d\bq\,\pi(\epsilon)\,d\epsilon
\label{eqn: marginal likelihood definition}
\end{equation}
where $\Delta_{m-1}$ is the $(m-1)$-dimensional simplex, and $\pi(\bq)$ and $\pi(\epsilon)$ are the Dirichlet and Beta densities in \eqref{eqn: nuisance priors}. Under \eqref{eqn: nuisance priors}, the marginal likelihood \eqref{eqn: marginal likelihood definition} is written in a closed form as follows.
\begin{proposition}\label{prop: collapsed likelihood}
Under \eqref{eqn: nuisance priors}, the marginal likelihood \eqref{eqn: marginal likelihood definition} is
\begin{equation}
\bar L(\theta)
=
\sum_{r_1=0}^{N_1}\cdots\sum_{r_m=0}^{N_m}
\Big\{
\prod_{j=1}^m
\binom{N_j}{r_j}
(\alpha_j)_{N_j-r_j}
\,p_j(\theta)^{r_j}
\Big\}
\frac{B(a+n-R,b+R)}{B(a,b)(\alpha_0)_{n-R}},
\label{eqn: collapsed likelihood expanded}
\end{equation}
where $R=\sum_{j=1}^m r_j$, $(x)_s=\Gamma(x+s)/\Gamma(x)$, and $B(x,y)=\Gamma(x)\Gamma(y)/\Gamma(x+y)$ is the beta function.
\end{proposition}
\noindent We refer to the likelihood \eqref{eqn: collapsed likelihood expanded} as the exact marginal likelihood.

Under the series representation given in \eqref{eqn: collapsed likelihood expanded}, each allocation $\br=(r_1,\ldots,r_m)$ assigns $r_j$ of the $N_j$ observations at atom~$z_j$ to the structural model and the remaining $N_j-r_j$ to contamination $Q$. The factor $\binom{N_j}{r_j}p_j(\theta)^{r_j}$ counts the structural probability over all choices of the $r_j$ observations assigned to the structural model at atom~$z_j$, the Dirichlet integral averages over the unknown contaminating distribution through $(\alpha_j)_{N_j-r_j}/(\alpha_0)_{n-R}$, and the Beta integral averages over $\epsilon$.

With the marginal likelihood \eqref{eqn: collapsed likelihood expanded}, the posterior for~$\theta$ is given by
\begin{equation}
\pi(\theta\mid y_{1:n})
=
\frac{\bar L(\theta)\,\pi(\theta)}
{\int_\Theta \bar L(\vartheta)\,\pi(\vartheta)\,d\vartheta} ,
\label{eqn: finite support posterior}
\end{equation}
where $\pi(\theta)$ is the prior density of $\theta$.
The choice of hyperparameters $\alpha_j$, $a$ and $b$ is discussed in Section~\ref{sec: theory}.

\subsection{Alternative representation of the exact marginal likelihood}
\label{sec: grouped representation}

After integrating out the contamination proportion and contaminating probability vector, the marginal likelihood \eqref{eqn: collapsed likelihood expanded} is no longer a product over observations. Grouping its allocations by $R$ gives the following form, which the large-sample results in Section~\ref{sec: theory} and the dynamic programming algorithm in Section~\ref{sec: computation} use.

\begin{proposition}\label{prop: multinomial form}
Under \eqref{eqn: nuisance priors}, with $\kappa_j(k)=(\alpha_j)_k/k!$ for $k\geq0$ and $\kappa_j(k)=0$ for $k<0$,
\begin{equation}
\bar L(\theta)
=
c_n
\sum_{R=0}^n
\frac{\Gamma(a+n-R)\,\Gamma(b+R)}{(\alpha_0)_{n-R}}\,
C_{n,R}(\theta),
\qquad
c_n=\frac{\prod_{j=1}^m N_j!}{B(a,b)\,\Gamma(a+b+n)},
\label{eqn: grouped likelihood}
\end{equation}
with 
\begin{equation}
C_{n,R}(\theta)
=
\sum_{\substack{r_1+\cdots+r_m=R\\0\leq r_j\leq N_j}}
\prod_{j=1}^m
\kappa_j(N_j-r_j)\frac{p_j(\theta)^{r_j}}{r_j!}=
\frac{1}{R!}
E_\theta\Big[\prod_{j=1}^m \Big\{\kappa_j(N_j-Z_j^{(R)})\,\mathbf{1}\{Z_j^{(R)}\leq N_j\}\Big\}\Big]
\label{eqn: C expectation}
\end{equation}
where $Z^{(R)}\sim\operatorname{Multinomial}\{R,\bp(\theta)\}$.
\end{proposition}

\section{Large-sample interpretation}
\label{sec: theory}

We then discuss the theoretical properties of the posterior \eqref{eqn: finite support posterior} under the exact marginal likelihood, especially under the large-sample regime. 
We consider two regimes for a prior on $\bq$. A uniform Dirichlet prior, $\alpha_j=1$, treats all contaminating distributions symmetrically, which is suitable when little is known about $Q$. On the other hand, a prior with smaller $\alpha_j$ values favors a contamination $Q$ on a few atoms, which is suitable for careless responses that occupy few values of the summary, such as straightlining and repeated response patterns \citep{WardMeade2023}, without knowing which values. Section~\ref{sec: uniform prior} considers the first regime and Section~\ref{sec: sparse prior} the second.
Asymmetric choices of $\alpha_j$ can encode external knowledge that some values, such as the values reached by straightlining, are more likely under careless responding.

In addition to $\balpha = (\alpha_1, \ldots, \alpha_m)$, $a$ and $b$ can encode external knowledge on $\epsilon$, such as careless proportions reported in comparable surveys \citep{WardMeade2023,ArthurEtAl2021}. Appendix~\ref{sec:supp-sensitivity} gives guidance on these values. For the theory, we assume $\epsilon\sim\operatorname{Beta}(1,b)$, which favors smaller values of $\epsilon$ for $b>1$. We write $\bar L_b=\bar L/c_n$ for \eqref{eqn: grouped likelihood} under this prior.
Throughout this section, we assume that the true population probabilities satisfy $p_{0j}>0$ for all $j=1,\ldots,m$ and that $N_j/n\to p_{0j}$, and write $\bp_0=(p_{01},\ldots,p_{0m})$.

\subsection{Uniform Dirichlet prior}
\label{sec: uniform prior}

Under $\alpha_j=1$ for all $j$, $\kappa_j\equiv1$ and $C_{n,R}(\theta)=\Pr_\theta\{Z_j^{(R)}\leq N_j,\ \forall j\}/R!$ in Proposition~\ref{prop: multinomial form}. The series representation \eqref{eqn: grouped likelihood} admits a direct large-sample interpretation.
We show that, for fixed finite $b$, the asymptotic likelihood after rescaling is determined by the smallest fraction $\epsilon$ required to represent $P_0$ in the form \eqref{eqn: huber model}. %

From \eqref{eqn: huber model}, a contaminating distribution with proportion $\epsilon>0$ exists if and only if $p_{0j}\geq(1-\epsilon)p_j(\theta)$ for all $j$, so we define the smallest such fraction
\begin{equation}
\epsilon_{\min}(\theta)
=
\inf\!\Big\{\epsilon\in[0,1]: p_{0j}\geq(1-\epsilon)\,p_j(\theta)\ \text{for all } j=1,\ldots,m\Big\},
\label{eqn: epsilon min definition}
\end{equation}
and write $\rho_0(\theta)=1-\epsilon_{\min}(\theta)$.
Equivalently, $\rho_0(\theta)$ is the largest fraction $\rho$ such that $p_{0j}\geq \rho p_j(\theta)$ for every atom.
The case $\rho_0(\theta)=1$ corresponds to exact structural fit, $P_0=P_\theta$.
For $\rho_0(\theta)<1$, define
\begin{equation}
A_{0,b}(\theta)
=
\int_0^{\rho_0(\theta)}
\rho^{b-1}(1-\rho)^{-(m-1)}\,d\rho .
\label{eqn: A zero b definition}
\end{equation}
Since the integrand does not depend on $\theta$, $A_{0,b}$ depends on $\theta$ only through $\rho_0(\theta)$, and a $\theta$ that needs a smaller contamination fraction to represent $P_0$ receives a larger value.
The following theorem shows the asymptotic behavior of the marginal likelihood under rescaling.
\begin{theorem}\label{thm: likelihood asymptotics}
Suppose $b>0$, $m\geq 2$, $N_j/n\to p_{0j}>0$ and $p_j(\theta)>0$ for all~$j$.
For every $\theta$ with $\rho_0(\theta)<1$,
\begin{equation}
n^{m-b-1}\bar L_b(\theta)
\to
\Gamma(m)\,A_{0,b}(\theta).
\label{eqn: likelihood limit}
\end{equation}
The convergence is uniform over any subset of the parameter space on which the atom probabilities are bounded away from zero and $\rho_0(\theta)$ is bounded away from one.
\end{theorem}

Theorem~\ref{thm: likelihood asymptotics} shows that in the misspecified regime, with $\rho_0(\theta) < 1$ for every $\theta$, the rescaled marginal likelihood has a limiting criterion obtained by integrating the same function over the interval $[0, \rho_0(\theta)]$. As $A_{0,b}(\theta)$ is strictly increasing in $\rho_0(\theta)$, the marginal likelihood asymptotically favors $\theta$ requiring a smaller contamination fraction to explain $P_0$. The factor $(1-\rho)^{-(m-1)}$ in \eqref{eqn: A zero b definition} becomes steeper as the support size $m$ grows, so larger finite supports can make the likelihood more concentrated at $\theta$ values with smaller minimum contamination fraction.

The case $\rho_0(\theta)=1$ is singular because the limiting integral diverges at $\rho = 1$. The following corollary shows the contraction result under correct specification.
\begin{corollary}\label{cor: exact fit dominance}
Suppose $b>0$, $m\geq 2$, $N_j/n\to p_{0j}>0$ for $j=1,\ldots,m$, and $\bp_0=\bp(\theta_0)$ for some $\theta_0\in\Theta$.
For every $\theta$ such that $\bp(\theta)\neq \bp(\theta_0)$ and $p_j(\theta)>0$ for all~$j$,
\begin{equation}
\frac{\bar L_b(\theta_0)}{\bar L_b(\theta)}
\to \infty .
\label{eqn: exact fit likelihood ratio}
\end{equation}
Moreover, for any set $K\subset\Theta$ on which $p_j(\theta)\geq c>0$ for all~$j$ and $\rho_0(\theta)\leq 1-\xi$ for some $c,\xi>0$,
\begin{equation}
\sup_{\theta\in K}
\frac{\bar L_b(\theta)}{\bar L_b(\theta_0)}
\to 0 .
\label{eqn: uniform exact fit likelihood ratio}
\end{equation}
\end{corollary}
In contrast to the contraction in Corollary \ref{cor: exact fit dominance}, Theorem~\ref{thm: likelihood asymptotics} implies that likelihood ratios between $\theta$ values need not diverge. This can be explained by the partial identifiability of $(\theta, \epsilon, Q)$ from $P_0$.

Next, we describe how the marginal likelihood orders structural explanations. For the standard likelihood posterior, misspecification leads to concentration near KL projections of $P_0$ onto the model under regularity conditions \citep{Berk1966,KleijnVanDerVaart2006}. For fixed finite $b$, the leading marginal-likelihood criterion is monotone in $\rho_0$ and therefore favors smaller values of $\epsilon_{\min}$ defined in \eqref{eqn: epsilon min definition}.

\begin{theorem}\label{thm: MLE property}
Suppose the population law satisfies $\bp_0=(1-\epsilon_0)\bp(\theta_0)+\epsilon_0 \bq$ for some $\epsilon_0\in[0,1]$ and probability vector~$\bq$.
Then $\epsilon_{\min}(\theta_0)\leq\epsilon_0$.  Moreover, if $\theta^*\in\arg\min_\theta\epsilon_{\min}(\theta)$, then
\begin{equation}
d_{\mathrm{TV}}\{\bp(\theta^*),\bp(\theta_0)\}
\leq 2\epsilon_0,
\label{eqn: tv stability}
\end{equation}
where $d_{\mathrm{TV}}(u,v)=\tfrac{1}{2}\sum_j|u_j-v_j|$ is the total variation distance.
\end{theorem}

Although a true structural parameter $\theta_0$ appears in the representation of $P_0$, this parameter is generally not identifiable under the Huber contamination model when $\epsilon_0>0$ \citep{Manski2003}. The proof uses the elementary fact that $\epsilon_{\min}(\theta)\leq\epsilon$ implies $d_{\mathrm{TV}}\{\bp(\theta),\bp_0\}\leq\epsilon$, followed by the triangle inequality. This bound can be translated into a parameter-level error bound in specific models. For example, when $P_\theta=\operatorname{Binomial}(M,\theta)$, the total-variation bound gives $|\theta^*-\theta_0|\leq 2\epsilon_0$. This bound holds for any contaminating distribution, so it is conservative and can be loose. The minimum-contamination set is determined by $\epsilon_{\min}$ and does not depend on $b$. The limiting posterior in Theorem~\ref{thm: posterior weak limit}, and finite-sample posterior summaries, can still depend on $b$ through $A_{0,b}$.

The posterior based on $\bar L_b$ quantifies uncertainty about $\theta$. As $n\to\infty$ under correct specification, Corollary~\ref{cor: exact fit dominance} gives likelihood-ratio dominance of the correctly specified structural law over alternatives bounded away from exact fit. Under genuine contamination, the posterior limit is non-degenerate under the conditions below. This behavior reflects the partial identification of $(\theta,\epsilon,\bq)$. The next theorem formalizes this asymptotic behavior under regularity conditions.

\begin{theorem}[Posterior limit under genuine contamination]\label{thm: posterior weak limit}
Let $b>0$ and suppose that $\Theta$ is compact, the prior has a continuous density $\pi(\theta)$ that is positive on $\Theta$, and $p_j(\theta)$ is continuous in $\theta$ for each $j$.
Assume further that there exist constants $c>0$ and $\xi>0$ such that $p_j(\theta)\geq c$ for all $j$ and all $\theta\in\Theta$, and $\rho_0(\theta)\leq 1-\xi$ for all $\theta\in\Theta$.
If $N_j/n\to p_{0j}>0$, then the marginal-likelihood posterior based on $\bar L_b$ converges weakly to the probability measure $\Pi_\infty$ with density
\begin{equation}
\pi_\infty(\theta)
=
\frac{A_{0,b}(\theta)\pi(\theta)}
{\int_\Theta A_{0,b}(\vartheta)\pi(\vartheta)\,d\vartheta}.
\label{eqn: posterior partial identification limit}
\end{equation}
\end{theorem}
When $\Theta$ is not compact, as for the model in Section~\ref{sec: real data}, Theorem~\ref{thm: posterior weak limit} applies with the prior restricted to a compact subset of parameters on which the atom probabilities are bounded from below.
\subsection{Sparse Dirichlet prior}
\label{sec: sparse prior}

Next, we consider a Dirichlet prior on $\bq$ with smaller $\alpha_j$ values, which favors a contaminating distribution with sparse support, without necessarily specifying which atoms are favored. 

Let $s(\theta)$ be the number of atoms with $p_{0j}=\rho_0(\theta)p_j(\theta)$, the atoms free of contamination when $P_0$ is written as \eqref{eqn: huber model} under the smallest contamination proportion $\epsilon_{\min}(\theta)$. Since $(\epsilon_0,Q)$ in \eqref{eqn: huber model} is not unique, we take $\epsilon_0=\epsilon_{\min}(\theta_0)$, so that at $\theta_0$ these are the atoms with $q_j=0$, and we write $s_0=s(\theta_0)$, so that $s_0=m$ when $\epsilon_0=0$. The following theorem shows that the posterior becomes consistent when $\theta_0$ leaves sufficiently more atoms free of contamination than any other $\theta$. We consider $\alpha_0\leq1$.

\begin{theorem}\label{thm: sparse contamination}
Let $b>0$ and $\alpha_0\leq1$, and let $Y_1,\ldots,Y_n$ be independent and identically distributed from $P_0=(1-\epsilon_0)P_{\theta_0}+\epsilon_0Q$ with $\epsilon_0\in[0,1)$ and $p_j(\theta_0)>0$ for all $j$. Then the following holds.
\begin{enumerate}
\item[(a)] If $p_j(\theta)>0$ for all $j$ and $s_0\geq 2s(\theta)+3$, then $\bar L_b(\theta_0)/\bar L_b(\theta)\to\infty$ in probability.
\item[(b)] Suppose that $\Theta\subset\mathbb{R}^d$ is compact, $\bp(\theta)$ is continuous with $p_j(\theta)\geq c>0$ on $\Theta$ and Lipschitz near the interior point $\theta_0$. If the prior density $\pi(\theta)$ is continuous with $\pi(\theta_0)>0$ and $s_0\geq2s(\theta)+d+3$ for all $\theta\neq\theta_0$, then the posterior is consistent at $\theta_0$, that is, $\Pi(U\mid Y_{1:n})\to1$ in probability for every open neighborhood $U$ of $\theta_0$.
\end{enumerate}
\end{theorem}

In practice, we take $\alpha_j=1/m$ which leads to $\alpha_0=1$. Without contamination, $s_0=m$, so part (b) gives consistency when $m\geq2s(\theta)+d+3$ for all $\theta\neq\theta_0$, which holds for the binomial model once $m\geq6$, since $p_j(\theta_0)/p_j(\theta)$ is monotone in $j$ and so $s(\theta)=1$.

\section{Posterior computation}
\label{sec: computation}

For the evaluation of the exact marginal likelihood $\bar L$ in \eqref{eqn: grouped likelihood} for any $\alpha_j$, $a$ and $b$, we deploy a dynamic programming approach. Direct summation of \eqref{eqn: collapsed likelihood expanded} requires
$\prod_{j=1}^m (N_j+1)$ terms, which becomes infeasible unless both $m$ and $n$ are small. The latent allocation representation gives an exact dynamic programming algorithm with computational cost $O\{n(n+m)\}$ once the atom probabilities $p_j(\theta)$ are available.

For each support point, define the cell polynomial
\begin{equation}
g_j(u;\theta)
=
\sum_{r=0}^{N_j}\frac{(\alpha_j)_{N_j-r}}{(N_j-r)!}\frac{p_j(\theta)^r}{r!}\,u^r .
\label{eqn: canonical cell polynomial}
\end{equation}
We write $C_j(R;\theta)$ for the coefficient of $u^R$ in
$\prod_{\ell=1}^j g_\ell(u;\theta)$.
The recursion starts from
\begin{equation}
C_0(0;\theta)=1,
\qquad
C_0(R;\theta)=0 \quad (R\geq 1),
\label{eqn: dp initialization}
\end{equation}
and proceeds by
\begin{equation}
C_j(R;\theta)
=
\sum_{r=0}^{\min\{R,N_j\}}
C_{j-1}(R-r;\theta)
\frac{(\alpha_j)_{N_j-r}}{(N_j-r)!}\frac{p_j(\theta)^r}{r!},
\label{eqn: dp recursion}
\end{equation}
for $j=1,\ldots,m$. Then
\begin{equation}
\bar L(\theta)
=
c_n\sum_{R=0}^n V_R\, C_m(R;\theta),
\qquad
V_R=
\frac{\Gamma(a+n-R)\,\Gamma(b+R)}{(\alpha_0)_{n-R}} .
\label{eqn: dp likelihood}
\end{equation}
Here $C_m(R;\theta)=C_{n,R}(\theta)$, and $c_n$ is as in Proposition~\ref{prop: multinomial form}. This sequential calculation only processes attainable partial totals, and $\bar L(\theta)$ can be obtained at arithmetic cost $O\{n(n+m)\}$, since step $j$ costs $O\{n(N_j+1)\}$ and $\sum_j N_j=n$, and one evaluation takes a few milliseconds for supports with hundreds of atoms (Appendix~\ref{sec:supp-computation}).

The same recursion gives the score of the marginal likelihood. Suppose $\theta=(\theta_1,\ldots,\theta_d)$ and the atom probabilities are differentiable, and write $D_{j\ell}(R;\theta)=\partial C_j(R;\theta)/\partial\theta_\ell$ with $D_{0\ell}(R;\theta)=0$. Differentiating \eqref{eqn: dp recursion} gives a forward recursion for $D_{j\ell}$ of the same order as the likelihood recursion. The resulting marginal-likelihood score is
\begin{equation}
\frac{\partial}{\partial \theta_\ell}
\log \bar L(\theta)
=
\frac{
\sum_{R=0}^n V_R D_{m\ell}(R;\theta)
}{
\sum_{R=0}^n V_R C_m(R;\theta)
}.
\label{eqn: collapsed score}
\end{equation}
Thus gradient-based posterior computation, such as Metropolis-adjusted Langevin or Hamiltonian Monte Carlo, does not require finite-difference approximation or grid enumeration. In implementation, the coefficient and derivative vectors should be normalized after each convolution step, with scale factors accumulated separately, to avoid numerical underflow and overflow. Appendix~\ref{sec:supp-computation} gives the full derivative recursion and the normalization details. The computation integrates out $\epsilon$, $\bq$, and all latent allocations, and none of these quantities is sampled.

In addition to $\theta$, posterior draws of $(\epsilon,\bq)$ can be obtained afterward. For a draw of $\theta$, the stored recursion gives a backward draw of the allocation $\br$, and given $\br$, $\epsilon\sim\operatorname{Beta}(a+n-R,b+R)$ and $\bq\sim\operatorname{Dirichlet}(\alpha_1+N_1-r_1,\ldots,\alpha_m+N_m-r_m)$.

\section{Simulation study}
\label{sec: simulation}

We investigate the performance of ExHuber in numerical experiments with item responses generated from a graded response model and with a bounded-count model.

The compared methods include the following.
\begin{itemize}
\item \textit{Naive}: the posterior under $P_\theta$ without any treatment of contamination.
\item \textit{Oracle}: the posterior under $P_\theta$ fitted to the engaged responses only, which is available only in simulations.
\item \textit{Coarsened}: the coarsened posterior of \citet{MillerDunson2019}, which raises the likelihood to a power $0.5$.
\item \textit{DPD}: the density-power-divergence posterior \citep{BasuEtAl1998,GhoshBasu2016} with power $\beta=0.5$ and unit learning rate.
\item \textit{Mixture IRT}: for item responses, the two-class mixture item response model of \citet{UglanovaEtAl2026}, with a unidimensional graded response class for attentive responses and a category-preference class for careless responses, chosen against a single-class model by the Bayesian information criterion and fitted with the R package \texttt{mirt} \citep{Chalmers2012}.
\item \textit{Longstring screening}: the naive posterior after removing respondents whose longest run of identical responses covers the whole scale (LS whole) or half of it (LS half).
\item \textit{ExHuber}: our posterior with $b=1$ and $\alpha_j=1/m$, whose only free quantities are the prior hyperparameters $b$ and $\alpha_j$ in place of a tuning constant, with their sensitivity in Appendix~\ref{sec:supp-sensitivity}.
\end{itemize}

\subsection{Item-level simulation}
\label{sec: item simulation}

We consider survey data in which $n$ respondents answer $J$ items on a five-point scale, where $X_{ij}\in\{1,\ldots,5\}$ is the response of respondent $i$ to item $j$. An engaged respondent with trait $\eta_i$ answers by the graded response model
\begin{equation}
\Pr(X_{ij}\geq k\mid\eta_i)=\operatorname{logit}^{-1}(a_j^\top\eta_i+\tau_{jk}),\qquad k=2,\ldots,5,\label{eqn: graded response model}
\end{equation}
where $\eta_i\in\mathbb{R}^D$ is the trait, $a_j\in\mathbb{R}^D$ is the discrimination of item $j$ with one nonzero entry, so that each item measures one trait, and $\tau_{jk}$ are the thresholds of item $j$, whose average sets how readily respondents choose high categories on that item. A negatively worded item has a negative discrimination. A proportion $\epsilon_0$ of respondents instead answer carelessly in one of the following ways.
\begin{itemize}
\item \textit{SL}: straightlining, with one category repeated on every item.
\item \textit{PATTERN}: a repeated sequence such as 1-2-3-4-5 or 1-5-1-5.
\item \textit{RCP}: a respondent-specific preference for some categories.
\item \textit{RR}: random responding, with a category drawn uniformly for each item.
\item \textit{MIX}: a mixture of SL, RCP and RR with weights sampled from $\operatorname{Dirichlet}(1,1,1)$.
\item \textit{SWITCH}: attentive responding that changes to random responding halfway through the items.
\end{itemize}
We analyze the sum score $Y_i\in\{0,\ldots,M\}$ of reverse-coded items with the beta-binomial model \citep{KeatsLord1962}
\begin{equation}
p_i \sim \operatorname{Beta}(\mu\phi,(1-\mu)\phi),
\qquad
Y_i \mid p_i \sim \operatorname{Binomial}(M, p_i),
\label{eqn: betabinom}
\end{equation}
which gives $p_y(\theta)=\binom{M}{y}B\{y+\mu\phi,\,M-y+(1-\mu)\phi\}/B\{\mu\phi,(1-\mu)\phi\}$ with $\theta=(\mu,\phi)$, where $\mu$ is the engaged mean and $\phi$ governs how much responses vary across respondents. The target is the engaged mean $\mu_0$ of $Y_i/M$. The sum score under \eqref{eqn: graded response model} is not exactly beta-binomial, so the case $\epsilon_0=0$ shows how this lack of fit affects each method.

The data are generated with $n=500$ and $J=10$, so that $M=40$, with each trait distributed as $N(0.75,1)$ and item parameters as in \citet{UglanovaEtAl2026}. The designs vary four factors of the engaged responses at two levels each, namely the range of the average threshold across items, $2$ or $4$, the magnitude of the discriminations, $0.5$ or $2$, the share of negatively worded items, $0\%$ or $40\%$, and the number of traits, one ($D=1$) or two with correlation $0.3$ ($D=2$), which gives $2^4=16$ designs. Each design has $\epsilon_0\in\{0,0.1,0.2,0.3,0.4\}$, and each experiment setup is replicated 100 times.

\begin{figure}[!t]
\centering
\includegraphics[width=\linewidth]{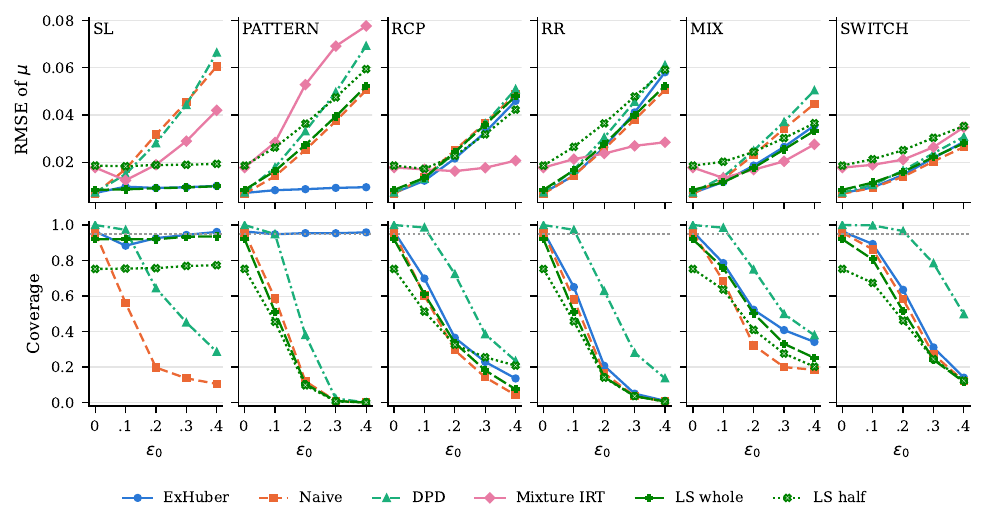}
\caption{Item-level simulation, averaged over 16 designs. Root mean squared error of the engaged mean $\mu$ (top) and coverage of the $95\%$ interval (bottom) against the contamination proportion $\epsilon_0$, for six careless types. LS whole and LS half fit the naive posterior after removing respondents with a longstring of the whole scale or of half of it. We report point estimates for mixture IRT. The dotted line marks $0.95$.}
\label{fig: s1 main}
\end{figure}

\begin{figure}[!t]
\centering
\includegraphics[width=\linewidth]{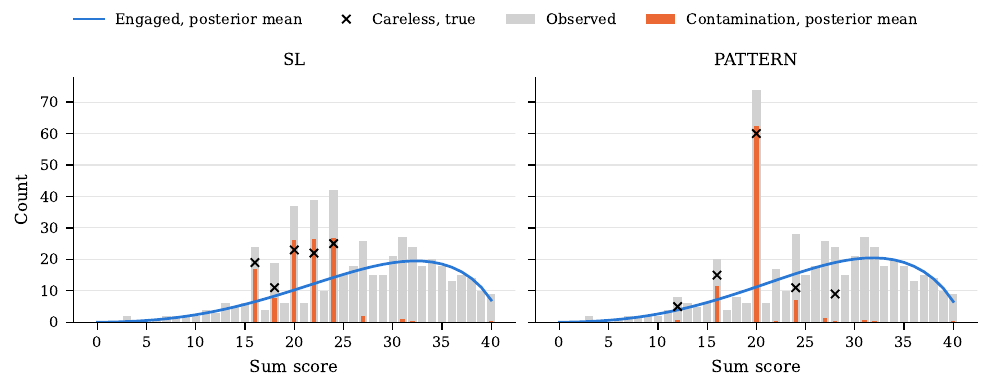}
\caption{Item-level simulation, one data set for SL and one for PATTERN at $\epsilon_0=0.2$, with the wider spread of the average threshold, strong discrimination and $40\%$ negatively worded items. Bars give the observed counts of the sum score and the posterior mean of the contamination counts $n\epsilon q_j$, the line gives the posterior mean of the engaged counts $n(1-\epsilon)p_j(\theta)$, and crosses give the true careless counts.}
\label{fig: q posterior}
\end{figure}

Figure~\ref{fig: s1 main} reports the error and coverage for $\mu_0$, averaged over the 16 designs. Without contamination, ExHuber and the naive posterior give nearly the same results, with coverage near the nominal level in every design, and the posterior mean of $\epsilon$ stays near zero (Appendix~\ref{sec:supp-sensitivity}).

When careless responses occupy few values of the sum score, as under SL and PATTERN, the ExHuber posterior stays accurate with coverage near the nominal level for every $\epsilon_0$, as the contamination is supported on a few values, in line with Theorem~\ref{thm: sparse contamination}. Additionally, as shown in Figure~\ref{fig: q posterior}, the posterior distribution of $Q$ captures the frequently observed contamination values from the data. On the other hand, the naive and DPD posteriors and mixture IRT show increasing error with $\epsilon_0$, and longstring screening only works under SL as intended. When careless sum scores overlap engaged ones, as under RR, RCP and MIX, all the methods show some error, and mixture IRT, which uses the item responses, has a smaller error from $\epsilon_0=0.2$.
Appendix~\ref{sec:supp-sensitivity} compares $\alpha_j=1/m$ with $\alpha_j=1$.

\subsection{Bounded-count model}
\label{sec: bounded count simulation}

We next consider data in which each response is a count $Y_i\in\{0,\ldots,M\}$, such as the number of endorsed items out of $M$. %
We consider the binomial model for the engaged counts $P_\theta=\operatorname{Binomial}(M,\theta)$ with atom probabilities $p_j(\theta)=\binom{M}{j}\theta^j(1-\theta)^{M-j}$ and $m=M+1$, and a proportion $\epsilon_0$ of counts instead follow $\operatorname{Binomial}(M,\theta_c)$. We consider the following two regimes, depending on how much $P_\theta$ and $Q$ overlap.
\begin{itemize}
\item \textit{Separated}: $\theta_c$ lies far from $\theta_0$, so $P_\theta$ and $Q$ place their mass on different counts.
\item \textit{Overlapping}: $\theta_c$ lies near $\theta_0$, so $P_\theta$ and $Q$ place their mass on similar counts.
\end{itemize}
The parameter of interest is $\theta_0$. With $\theta\sim\operatorname{Uniform}(0,1)$, the naive and oracle posteriors are beta distributions, and the other posteriors are computed by direct integration on a fine grid over $\theta$.

We take $n=300$, $M=20$ and $\theta_0=0.30$, with $\theta_c=0.75$ in the separated regime and $\theta_c=0.45$ in the overlapping regime, and $\epsilon_0\in\{0,0.1,0.2,0.3,0.4\}$. Each setting has 100 replications. Here we show DPD with $\beta=0.25$ and $\beta=1$.

\begin{figure}[!t]
\centering
\includegraphics[width=0.72\linewidth]{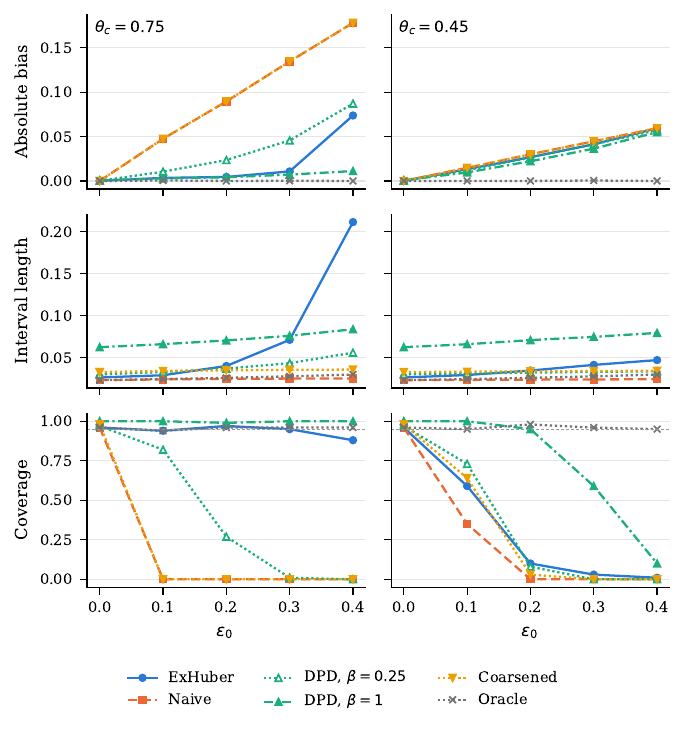}
\caption{Bounded-count model. Absolute bias (top), interval length (middle) and coverage (bottom) for $\theta$ against the contamination proportion $\epsilon_0$, in a separated regime ($\theta_c=0.75$, left) and an overlapping regime ($\theta_c=0.45$, right). Each point averages over $100$ replications with $n=300$, $M=20$ and $\theta_0=0.30$. The dotted line marks $0.95$. DPD is shown at $\beta=0.25$ and $\beta=1$.}
\label{fig: binom-calib}
\end{figure}

Figure~\ref{fig: binom-calib} reports bias, interval length and coverage for $\theta_0$. In the separated regime, the ExHuber posterior keeps a small bias and coverage near the nominal level up to $\epsilon_0=0.3$, %
while the naive and coarsened posteriors miss $\theta_0$. Its interval becomes wider as $\epsilon_0$ increases, which reflects the uncertainty about how much of the data is contaminated. Additionally, the performance of DPD significantly depends on $\beta$. With $\beta=0.25$, it misses $\theta_0$ under contamination, and with $\beta=1$ it covers $\theta_0$ at the cost of longer intervals even without contamination. In the overlapping regime, $P_\theta$ and $Q$ give similar counts, so every method that uses all responses is biased toward $\theta_c$, and no prior choice for ExHuber changes this (Appendix~\ref{sec:supp-sensitivity}).

\section{Big Five scales}
\label{sec: real data}

We consider the crowdsourced careless-responding samples of \citet{BruhlmannPetralitoAeschbachOpwis2020}, in which participants completed the Big Five Inventory on a five-point scale. The inventory has five scales, which measure extraversion (EX), agreeableness (AG), conscientiousness (CON), neuroticism (NEU) and openness (OP). For each scale with $J$ items, as in Section~\ref{sec: item simulation}, $Y_i\in\{0,\ldots,M\}$ is the sum score of the reverse-coded items minus $J$, with $M=4J$, and the target is the engaged mean of $Y_i/M$. The survey also includes two attention checks, an instructed-response item and a bogus item. As the ground truth of the population is unknown, we use the respondents who pass both checks as an attention-clean reference.

We take the beta-binomial model \eqref{eqn: betabinom} as the structural model, as in Section~\ref{sec: item simulation}. We compare with the naive and coarsened posteriors, DPD and mixture IRT of Section~\ref{sec: simulation}.

The data include $n=394$ participants and 44 items, with $J=8,9,9,8,10$ items on EX, AG, CON, NEU and OP, and about a third of the respondents fail at least one attention check. We use $\mu\sim\operatorname{Beta}(1,1)$ and $\log\phi\sim N(\log 10,1)$ independently, with $b=1$ and $\alpha_j=1/m$, as in Section~\ref{sec: simulation}, so that the prior on $\epsilon$ is uniform and that on $q$ favors few atoms (Section~\ref{sec: sparse prior}).

\begin{figure}[!t]
\centering
\includegraphics[width=\linewidth]{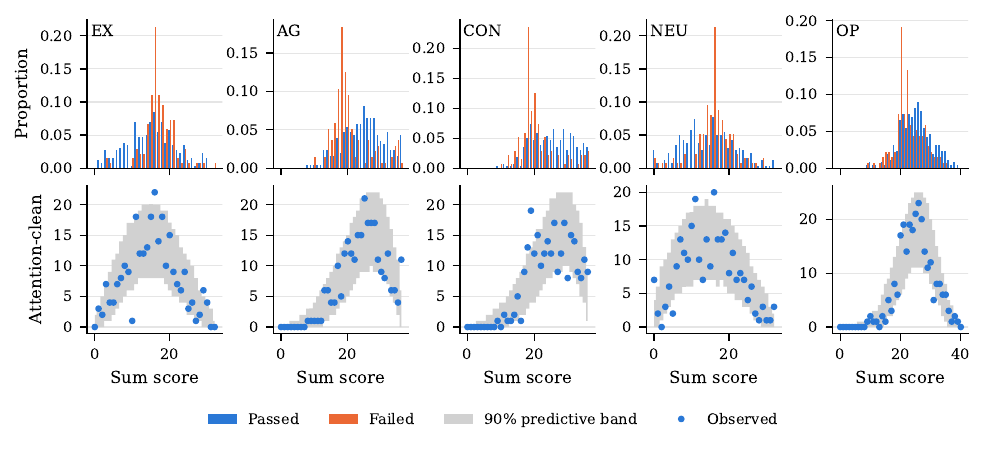}
\caption{Big Five scales. The top row shows the proportion of each sum score among respondents who pass and among those who fail the attention checks, and the bottom row shows the attention-clean counts with $90\%$ pointwise posterior predictive bands from the beta-binomial fit to these counts.}
\label{fig: bfi}
\end{figure}

\begin{table}[!t]
\centering
\caption{Engaged mean $\mu$ on each scale ($n=394$), with $95\%$ posterior intervals, and the proportion $\epsilon$ of contamination. The reference fits the 258 respondents who pass the attention checks and is used by no estimator. Coarsened posterior means equal the naive ones to two decimals. DPD is the density-power-divergence posterior.}
\label{tab: bfi}
\small
\setlength{\tabcolsep}{4pt}
\begin{tabular}{lcccccc}
\hline
Scale & Reference & Naive & DPD & Mixture IRT & ExHuber $\mu$ & ExHuber $\epsilon$ \\
\hline
EX & 0.477 & 0.496 & 0.504 & 0.482 & $0.498\ [0.475,0.520]$ & $0.15\ [0.08,0.25]$ \\
AG & 0.688 & 0.652 & 0.627 & 0.677 & $0.645\ [0.617,0.676]$ & $0.15\ [0.06,0.24]$ \\
CON & 0.710 & 0.668 & 0.641 & 0.702 & $0.701\ [0.663,0.738]$ & $0.28\ [0.20,0.37]$ \\
NEU & 0.457 & 0.468 & 0.478 & 0.457 & $0.475\ [0.454,0.495]$ & $0.12\ [0.07,0.17]$ \\
OP & 0.629 & 0.607 & 0.596 & 0.616 & $0.614\ [0.600,0.629]$ & $0.08\ [0.04,0.14]$ \\
\hline
\end{tabular}
\end{table}

Respondents who fail the attention checks concentrate around
the midpoint of each scale, while those who pass spread over the scale (Figure~\ref{fig: bfi}). A posterior predictive check on the attention-clean counts shows that the beta-binomial follows the overall shape of each scale, with departures near the midpoint, at the scale maximum on AG and at the minimum on NEU. As in Section~\ref{sec: collapsed likelihood}, these departures enter $Q$, so the posterior of $\epsilon$ also reflects lack of fit of $P_\theta$.

Table~\ref{tab: bfi} reports the estimates of the engaged mean $\mu$ and the contamination proportion $\epsilon$ on each scale. ExHuber is closer to the reference than DPD on every scale. On CON, it is as close to the reference as mixture IRT without specifying the form of careless responses. On the other scales, it stays near the naive value, and on every scale the posterior mean of $\epsilon$ lies below the proportion of respondents who fail at least one attention check, under each prior on $\epsilon$ in Table~\ref{tab: bfi eps prior}. As in Figure~\ref{fig: bfi}, respondents who fail the checks overlap those who pass over most of each scale. Since the prior with $\alpha_j=1/m$ favors $Q$ on a few values (Section~\ref{sec: sparse prior}), careless answers spread over the scale are explained by $P_\theta$, and $Q$ captures only concentrated departures such as the excess near the midpoint, as under RR, RCP and MIX in Section~\ref{sec: item simulation} (Figure~\ref{fig: s1 eps}). Therefore, the posterior samples of $\epsilon$ should be interpreted as the proportion of departures from $P_\theta$ rather than of all careless responding. 

\section{Conclusion}
\label{sec: discussion}
Careless responses take forms that are hard to specify in advance. Mixture models specify their form, and screening rules remove respondents by a cutoff. ExHuber places priors on the contamination proportion and on the contaminating distribution of the Huber model on a finite support and integrates both out, which gives an exact marginal likelihood for the structural parameter without specifying the form of careless responses. It replaces the tuning constants of generalized Bayesian methods, such as coarsening levels, divergence powers and learning rates, with prior hyperparameters that encode the expected careless proportion and the number of values careless responses occupy, and dynamic programming computes the likelihood at a cost of $O\{n(n+m)\}$. Under a Dirichlet prior favoring the contaminating distribution on a few atoms, the posterior of $\theta$ is consistent when careless responses occupy few values, without knowing which values. Simulations from item response models show accurate inference under straightlining and repeated response patterns, and the application to Big Five scales estimates the engaged mean of each scale.

ExHuber and mixture item response models use different information. A mixture model uses the item responses, so it can detect inconsistency across items, and it gives each respondent a posterior class probability \citep{UglanovaEtAl2026}. When careless scores overlap engaged ones, a mixture model gains from this information (Section~\ref{sec: item simulation}). ExHuber uses a scale score with an arbitrary $Q$ and a structural model for the score, and it suits straightlining and repeated response patterns. A questionnaire with many items can be handled by analyzing one scale score at a time, as in Section~\ref{sec: real data}, and a large support can be reduced by merging adjacent scores.

The marginal likelihood extends to a covariate-dependent structural model (Remark~\ref{rem: non-iid}), and its computation and theory in that setting, dependent data and large supports are left to future work.

\section*{Data availability}
The sum scores analyzed in the application are derived from the openly available careless-responding dataset of \citet{BruhlmannPetralitoAeschbachOpwis2020}. Code and derived data needed to reproduce all analyses, figures, and tables will be made available upon publication.

\section*{Financial support}
This research received no specific grant from any funding agency, commercial, or not-for-profit sectors.

\section*{Competing interests}
The author declares no competing interests.

\section*{Use of generative artificial intelligence}
During the preparation of this manuscript in 2026, the author used Claude Opus~4 (Anthropic), accessible at https://claude.ai, for refining the language and grammar. The tool was not used to produce the study's methods, results, or conclusions. The author reviewed, edited, and verified all such text and takes full responsibility for the content of the publication.

\appendix
\renewcommand{\theequation}{S\arabic{equation}}
\setcounter{equation}{0}

\section{Proofs}
\label{sec:supp-proofs}

Throughout the proofs, $c$ and $C$ denote positive constants that do not depend on $n$ and may change in different context.

\begin{proof}[Proof of Proposition \ref{prop: collapsed likelihood}]
Starting from the likelihood expression in the main text \eqref{eqn: finite support likelihood}, expand the factor at each support point $z_j$ by the binomial theorem,
\begin{equation}
\{(1-\epsilon)p_j(\theta)+\epsilon q_j\}^{N_j}
=
\sum_{r_j=0}^{N_j}
\binom{N_j}{r_j}
p_j(\theta)^{r_j}
(1-\epsilon)^{r_j}
\epsilon^{N_j-r_j}
q_j^{N_j-r_j},
\label{eqn: appendix binomial expansion}
\end{equation}
where $r_j$ is the number of the $N_j$ observations at atom $z_j$ attributed to the structural model, while $N_j-r_j$ are attributed to contamination.
Multiplying this expansion over all atoms $j=1,\ldots,m$, and collecting powers of $\epsilon$ and $1-\epsilon$ using $R=\sum_{j=1}^m r_j$, a generic term in the expanded product can be written as
\begin{equation}
\Big\{
\prod_{j=1}^m
\binom{N_j}{r_j}
p_j(\theta)^{r_j}
q_j^{N_j-r_j}
\Big\}
(1-\epsilon)^R\epsilon^{n-R}.
\end{equation}
We integrate the components involving $\bq$ under the $\operatorname{Dirichlet}(\alpha_1, \ldots, \alpha_m)$ prior over the simplex $\Delta_{m-1}$,
\begin{equation}
\begin{aligned}
&\int_{\Delta_{m-1}}
\prod_{j=1}^m q_j^{N_j-r_j}
\frac{1}{B(\alpha)}
\prod_{j=1}^m q_j^{\alpha_j-1}\,d\bq  \\
&\qquad =
\frac{\Gamma(\alpha_0)}{\prod_{j=1}^m \Gamma(\alpha_j)}
\frac{\prod_{j=1}^m \Gamma(\alpha_j + N_j-r_j)}
{\Gamma(\alpha_0 + n-R)}
=
\frac{\prod_{j=1}^m(\alpha_j)_{N_j-r_j}}{(\alpha_0)_{n-R}},
\end{aligned}
\label{eqn: appendix dirichlet integral}
\end{equation}
where $\alpha_0=\sum_{j=1}^m\alpha_j$ and $B(\alpha)=\prod_{j=1}^m\Gamma(\alpha_j)/\Gamma(\alpha_0)$ is the multivariate beta function. Next, integrating the remaining factors involving $\epsilon$ against the $\operatorname{Beta}(a,b)$ prior over the interval $[0,1]$ yields the beta integral
\begin{equation}
\int_0^1
\epsilon^{n-R}(1-\epsilon)^R
\frac{\epsilon^{a-1}(1-\epsilon)^{b-1}}{B(a,b)}\,d\epsilon
=
\frac{B(a+n-R,b+R)}{B(a,b)}.
\label{eqn: appendix beta integral}
\end{equation}
Combining \eqref{eqn: appendix binomial expansion}--\eqref{eqn: appendix beta integral} and summing over all valid indices $r_1, \ldots, r_m$ establishes \eqref{eqn: collapsed likelihood expanded}.
\end{proof}

\begin{proof}[Proof of Proposition \ref{prop: multinomial form}]
Since $\binom{N_j}{r_j}(\alpha_j)_{N_j-r_j}=N_j!\,\kappa_j(N_j-r_j)/r_j!$ and $B(a+n-R,b+R)=\Gamma(a+n-R)\Gamma(b+R)/\Gamma(a+b+n)$, each term of \eqref{eqn: collapsed likelihood expanded} equals
\begin{equation}
\frac{\prod_{j=1}^m N_j!}{B(a,b)\,\Gamma(a+b+n)}\,
\frac{\Gamma(a+n-R)\,\Gamma(b+R)}{(\alpha_0)_{n-R}}
\prod_{j=1}^m\kappa_j(N_j-r_j)\frac{p_j(\theta)^{r_j}}{r_j!}.
\end{equation}
The first factor is $c_n$ and the second depends on $\br$ only through $R$, so grouping the terms by $R$ gives \eqref{eqn: grouped likelihood}, with $C_{n,R}(\theta)$ the sum of the last product over $\br$ with $r_1+\cdots+r_m=R$ and $0\leq r_j\leq N_j$. For a nonnegative integer vector $\br$ summing to $R$, $\Pr_\theta\{Z^{(R)}=\br\}=R!\prod_{j=1}^m p_j(\theta)^{r_j}/r_j!$, so $C_{n,R}(\theta)$ is $1/R!$ times the expectation in \eqref{eqn: C expectation}.
\end{proof}

We use two lemmas in the proofs below. The first gives the marginal likelihood for $a=1$ used in Section~\ref{sec: theory}, and the second is a binomial tail bound.

\begin{lemma}\label{lem: representation}
Under \eqref{eqn: nuisance priors} with $a=1$,
\begin{equation}
\bar L_b(\theta)
=
\sum_{R=0}^n
W_{n,R}\,
E_\theta\Big[\prod_{j=1}^m \kappa_j(N_j-Z_j^{(R)})\,\mathbf{1}\{Z_j^{(R)}\leq N_j,\ \forall j\}\Big],
\qquad
W_{n,R}=\frac{\Gamma(n-R+1)\,\Gamma(b+R)}{(\alpha_0)_{n-R}\,R!},
\label{eqn: marginal likelihood homogeneous}
\end{equation}
where $Z^{(R)}\sim\operatorname{Multinomial}\{R,\bp(\theta)\}$. Moreover, $\kappa_j(k)\asymp(k+1)^{\alpha_j-1}$ over $k\geq0$ and $W_{n,R}\asymp(R+1)^{b-1}(n-R+1)^{1-\alpha_0}$ over $0\leq R\leq n$.
\end{lemma}

\begin{proof}
With $a=1$, $\Gamma(a+n-R)=\Gamma(n-R+1)$, so dividing \eqref{eqn: grouped likelihood} by $c_n$ and using \eqref{eqn: C expectation} gives \eqref{eqn: marginal likelihood homogeneous}. For fixed $s$ and $t$, $\Gamma(x+s)/\Gamma(x+t)=x^{s-t}\{1+O(x^{-1})\}$ as $x\to\infty$. Applied to $\kappa_j(k)=\Gamma(\alpha_j+k)/\{\Gamma(\alpha_j)\Gamma(k+1)\}$, to $\Gamma(n-R+1)/(\alpha_0)_{n-R}=\Gamma(\alpha_0)\Gamma(n-R+1)/\Gamma(\alpha_0+n-R)$ and to $\Gamma(b+R)/R!$, this gives the two rates, since each ratio is positive and finite for small arguments.
\end{proof}

\begin{lemma}\label{lem: binomial tail}
Let $Z\sim\operatorname{Binomial}(R,p)$ with $R\leq n$, and let $\gamma>0$. If $N\geq Rp+\gamma n$, then $\Pr(Z>N)\leq e^{-2\gamma^2n}$, and if $N\leq Rp-\gamma n$, then $\Pr(Z\leq N)\leq e^{-2\gamma^2n}$.
\end{lemma}

\begin{proof}
For $R\geq1$ and $t>0$, Hoeffding's inequality gives $\Pr(Z-Rp\geq t)\leq e^{-2t^2/R}$ and $\Pr(Z-Rp\leq -t)\leq e^{-2t^2/R}$. Taking $t=\gamma n$ and using $R\leq n$ gives both bounds. For $R=0$, both probabilities are zero.
\end{proof}

\begin{proof}[Proof of Theorem \ref{thm: likelihood asymptotics}]
By Lemma~\ref{lem: representation} with $\alpha_j=1$, so that $\kappa_j\equiv1$ and $\alpha_0=m$,
\begin{equation}
\bar L_b(\theta)
=
\sum_{R=0}^n
W_{n,R}
\Pr_\theta\{Z_j^{(R)}\leq N_j,\ j=1,\ldots,m\},
\label{eqn: proof W likelihood}
\end{equation}
where the weights become
\begin{equation}
W_{n,R}
=
\frac{\Gamma(n-R+1)\Gamma(b+R)}
{(m)_{n-R}R!}
=
\Gamma(m)\frac{\Gamma(n-R+1)\Gamma(b+R)}{\Gamma(n-R+m)\Gamma(R+1)}.
\label{eqn: proof W definition}
\end{equation}

To establish the limit of $n^{m-b-1}\bar L_b(\theta)$, we use Stirling's asymptotic expansion for the gamma function ratio, which states that for any fixed constants $s, t$, $\Gamma(x+s)/\Gamma(x+t) = x^{s-t}(1 + O(x^{-1}))$ as $x \to \infty$. Applying this to the weights with $\rho = R/n$, we have
\begin{equation}
\frac{\Gamma(n-R+1)}{\Gamma(n-R+m)} = (n-R)^{1-m}\left(1 + O\left(\frac{1}{n-R}\right)\right), \quad \frac{\Gamma(b+R)}{\Gamma(R+1)} = R^{b-1}\left(1 + O\left(\frac{1}{R}\right)\right).
\end{equation}
Consequently, we can write the rescaled weight $n^{m-b} W_{n,R}$ as
\begin{equation}
n^{m-b} W_{n,R} = \Gamma(m) \rho^{b-1}(1-\rho)^{-(m-1)} (1 + E_{n,R}),
\label{eqn: proof W stirling}
\end{equation}
where the error term satisfies $|E_{n,R}| \leq \frac{C}{n\rho(1-\rho)}$ for a constant $C > 0$ depending only on $m$ and $b$. Thus, the convergence is uniform for $\rho$ in any compact subinterval of $(0,1)$.

We fix $\theta$ such that $\rho_0(\theta) < 1$. Then we set $\eta > 0$ satisfying $\eta < \min(\rho_0(\theta), 1-\rho_0(\theta))$ and $\delta > 0$ satisfying $\delta < \eta$, and decompose the sum in \eqref{eqn: proof W likelihood} into four distinct regions based on the value of $\rho = R/n$,
\begin{align*}
I_1 &= \{R : 0 \leq R < \delta n\}, \\
I_2 &= \{R : \delta n \leq R \leq (\rho_0(\theta) - \eta)n\}, \\
I_3 &= \{R : (\rho_0(\theta) - \eta)n < R < (\rho_0(\theta) + \eta)n\}, \\
I_4 &= \{R : (\rho_0(\theta) + \eta)n \leq R \leq n\}.
\end{align*}

First, consider $R \in I_2$. For these values, $\rho$ lies in a compact subset of $(0,1)$. By the definition of $\rho_0(\theta) = 1 - \epsilon_{\min}(\theta)$, we have $p_{0j} \geq \rho_0(\theta) p_j(\theta) \geq (\rho + \eta) p_j(\theta)$ for all $j=1,\ldots,m$. Since $N_j/n \to p_{0j}$, for sufficiently large $n$, we have $N_j/n \geq (\rho + \eta/2)p_j(\theta)$, which implies $N_j \geq R p_j(\theta) + \frac{\eta}{2} n p_j(\theta)$. Since each multinomial marginal $Z_j^{(R)}$ is $\operatorname{Binomial}\{R,p_j(\theta)\}$, Lemma~\ref{lem: binomial tail} applied with $\eta c/2$ in place of its $\gamma$, where $c=\min_j p_j(\theta)>0$, gives $\Pr_\theta\{Z_j^{(R)}>N_j\}\leq\exp(-\gamma n)$ with $\gamma=\eta^2c^2/2$. By a union bound, $1 - \Pr_\theta\{Z_j^{(R)}\leq N_j,\ \forall j\} \leq m \exp(-\gamma n)$. This exponential convergence to 1 is uniform for $R \in I_2$. Therefore, the sum over $I_2$ satisfies
    \begin{align*}
&n^{m-b-1}
\sum_{R \in I_2}
W_{n,R}\Pr_\theta\{Z_j^{(R)}\leq N_j,\ \forall j\} \\
&\quad =
\frac{1}{n} \sum_{R \in I_2}
\Gamma(m) \rho^{b-1}(1-\rho)^{-(m-1)}
\{1 + O(n^{-1})\}\{1 - O(e^{-\gamma n})\} \\
&\quad \to
\Gamma(m)
\int_\delta^{\rho_0(\theta)-\eta}
\rho^{b-1}(1-\rho)^{-(m-1)}\,d\rho
\quad \text{as } n \to \infty,
\end{align*}
recognized as a Riemann sum converging to the specified integral.

Second, consider $R \in I_4$. By definition of $\rho_0(\theta)$, since $\rho \geq \rho_0(\theta) + \eta$, there exists at least one coordinate $j$ such that $\rho p_j(\theta) - p_{0j} \geq \eta p_j(\theta)$. Given $N_j/n \to p_{0j}$, for large $n$ we have $N_j/n \leq \rho p_j(\theta) - \frac{\eta}{2} p_j(\theta)$, so $N_j \leq R p_j(\theta) - \frac{\eta}{2} n p_j(\theta)$. Lemma~\ref{lem: binomial tail} then yields
\begin{equation}
\Pr_\theta\{Z_j^{(R)} \leq N_j, \ \forall j\} \leq \Pr_\theta\{Z_j^{(R)} \leq N_j\} \leq \exp(-\gamma n).
\end{equation}
Since $n^{m-b} W_{n,R}$ is bounded by a polynomial in $n$ over $I_4$, the entire sum over $I_4$ vanishes exponentially fast
\begin{equation}
n^{m-b-1} \sum_{R \in I_4} W_{n,R} \Pr_\theta\{Z_j^{(R)}\leq N_j,\ \forall j\} \leq \exp(-\gamma n) \frac{1}{n} \sum_{R \in I_4} n^{m-b} W_{n,R} \to 0.
\end{equation}

Third, for the boundary band $I_3$, we bound the multinomial probability by 1. Since $n^{m-b} W_{n,R}$ is uniformly bounded by a constant $M_{\eta}$ on this region, with $M_\eta\leq M_{\bar\eta}$ for $\eta\leq\bar\eta$ because the region shrinks with $\eta$, the sum satisfies
\begin{equation}
n^{m-b-1} \sum_{R \in I_3} W_{n,R} \Pr_\theta\{Z_j^{(R)}\leq N_j,\ \forall j\} \leq \frac{1}{n} \sum_{R \in I_3} M_{\eta} \leq 2\eta M_{\eta},
\end{equation}
which can be made arbitrarily small as $\eta \downarrow 0$.

Finally, for the lower endpoint region $I_1$, bound the multinomial probability by one. The term $R=0$ contributes $O(n^{-b})$ after multiplication by $n^{m-b-1}$. For $1\leq R<\delta n$, the rate of $W_{n,R}$ in Lemma~\ref{lem: representation} yields
\begin{equation}
n^{m-b} W_{n,R}\leq C'\left(\frac{R}{n}\right)^{b-1}.
\end{equation}
Therefore,
\begin{equation}
n^{m-b-1} \sum_{R \in I_1} W_{n,R} \Pr_\theta\{Z_j^{(R)}\leq N_j,\ \forall j\}
\leq
O(n^{-b})+\frac{1}{n}\sum_{1\leq R<\delta n} C'\left(\frac{R}{n}\right)^{b-1},
\end{equation}
which is bounded by a constant multiple of $\int_0^\delta \rho^{b-1}\,d\rho$ and hence vanishes as $\delta \downarrow 0$ since $b > 0$. Combining the limits of the four regions by first taking $n \to \infty$, then $\delta \downarrow 0$ and $\eta \downarrow 0$, we conclude that
\begin{equation}
n^{m-b-1}\bar L_b(\theta)
\to
\Gamma(m)
\int_0^{\rho_0(\theta)}
\rho^{b-1}(1-\rho)^{-(m-1)}\,d\rho
=
\Gamma(m)A_{0,b}(\theta).
\end{equation}

For the uniform statement, consider a subset $K$ of $\Theta$ where $p_j(\theta)\geq c>0$ for all $j$ and $\rho_0(\theta)\leq 1-\xi$ for some $\xi>0$. Compactness of $K$ is not needed, since every bound above depends on $\theta$ only through $c$ and $\xi$. The Stirling approximation error bound $|E_{n,R}| \leq \frac{C}{n\rho(1-\rho)}$ is uniform since $\rho\leq1-\xi/2$. Since $p_j(\theta)\leq1$, $\rho_0(\theta)\geq\min_jp_{0j}>0$, so a single $\eta<\min(\min_jp_{0j},\xi)/2$ serves every $\theta\in K$. The multinomial tail bounds are also uniform because the separation from the boundary is at least $c\eta + o(1)$, which is bounded away from zero. This establishes the asserted uniform convergence.
\end{proof}

\begin{proof}[Proof of Corollary \ref{cor: exact fit dominance}]
    Since $\bp_0=\bp(\theta_0)$, the minimum contamination fraction satisfies $\epsilon_{\min}(\theta_0)=0$, which implies $\rho_0(\theta_0)=1$. If $\bp(\theta)\neq \bp(\theta_0)$ and $p_j(\theta)>0$ for all $j$, then at least one coordinate must satisfy $p_j(\theta)>p_{0j}$, because both $\bp(\theta)$ and $\bp(\theta_0)$ are probability vectors summing to one. This implies $\epsilon_{\min}(\theta)>0$ and consequently $\rho_0(\theta)<1$. By Theorem~\ref{thm: likelihood asymptotics}, the rescaled marginal likelihood converges to a finite positive constant
\begin{equation}
n^{m-b-1}\bar L_b(\theta)
\to
\Gamma(m)A_{0,b}(\theta).
\end{equation}
To prove the pointwise likelihood-ratio assertion, it suffices to show that the same rescaling diverges to infinity at $\theta_0$. Fix any $\delta, \eta$ such that $0 < \delta < 1-\eta < 1$. Retaining only the terms from $R = \delta n$ to $R = (1-\eta)n$ in \eqref{eqn: proof W likelihood}, we have the lower bound
\begin{equation}
n^{m-b-1}\bar L_b(\theta_0)
\geq
n^{m-b-1}
\sum_{\delta n\leq R\leq (1-\eta)n}
W_{n,R}
\Pr_{\theta_0}\{Z_j^{(R)}\leq N_j,\ j=1,\ldots,m\}.
\label{eqn: proof exact fit lower bound}
\end{equation}
For any $\rho = R/n \leq 1-\eta$, the multinomial expectation satisfies $E_{\theta_0}\{Z_j^{(R)}\} = R p_{0j} \leq (1-\eta)n p_{0j}$, while the observed counts satisfy $N_j/n \to p_{0j} > 0$. Thus, $N_j - E_{\theta_0}\{Z_j^{(R)}\} \geq \eta n p_{0j} + o(n) \geq \frac{\eta}{2} n p_{0j} > 0$ for large $n$. By Lemma~\ref{lem: binomial tail} and a union bound, the multinomial probability in \eqref{eqn: proof exact fit lower bound} converges to one uniformly for $\rho \in [\delta, 1-\eta]$. Utilizing the uniform Stirling approximation \eqref{eqn: proof W stirling} on this interval, the sum behaves as a Riemann sum, yielding
\begin{equation}
\liminf_{n\to\infty}
n^{m-b-1}\bar L_b(\theta_0)
\geq
\Gamma(m)
\int_\delta^{1-\eta}
\rho^{b-1}(1-\rho)^{-(m-1)}\,d\rho .
\end{equation}
Since $m\geq2$, the exponent of the $(1-\rho)$ term in the denominator is $m-1 \geq 1$, which makes the integrand non-integrable at the upper limit $\rho=1$. Taking the limit as $\eta \downarrow 0$, the integral diverges to infinity, which proves that $n^{m-b-1}\bar L_b(\theta_0)\to\infty$. Consequently, the ratio satisfies
\begin{equation}
\frac{\bar L_b(\theta_0)}{\bar L_b(\theta)}
=
\frac{n^{m-b-1}\bar L_b(\theta_0)}
{n^{m-b-1}\bar L_b(\theta)}
\to\infty ,
\end{equation}
proving the pointwise statement. For the uniform assertion over a set $K \subset \Theta$ where $\rho_0(\theta) \leq 1-\xi$, Theorem~\ref{thm: likelihood asymptotics} guarantees that $\sup_{\theta\in K} n^{m-b-1}\bar L_b(\theta) = O(1)$. Combining this uniform upper bound with the divergence of the numerator at $\theta_0$ yields $\sup_{\theta\in K} \frac{\bar L_b(\theta)}{\bar L_b(\theta_0)} \to 0$.
\end{proof}

\begin{proof}[Proof of Theorem \ref{thm: MLE property}]
    The assumed Huber population representation states that $p_{0j} = (1-\epsilon_0)p_j(\theta_0) + \epsilon_0 q_j$ for all $j=1,\ldots,m$. Since $q_j \geq 0$, it follows that $p_{0j} \geq (1-\epsilon_0)p_j(\theta_0)$ for all $j$. This implies that $\epsilon = \epsilon_0$ is a feasible element in the set defined in \eqref{eqn: epsilon min definition}, and by taking the infimum, we have $\epsilon_{\min}(\theta_0)\leq\epsilon_0$.

Next, suppose that $\theta \in \Theta$ satisfies $\epsilon_{\min}(\theta)\leq\epsilon$ for some $\epsilon \in [0,1]$. By definition, this implies that $p_{0j}\geq (1-\epsilon)p_j(\theta)$ for all $j=1,\ldots,m$. If $\epsilon>0$, we can define the vector $\bq^\theta$ elementwise by
\begin{equation}
q_j^\theta
=
\frac{p_{0j}-(1-\epsilon)p_j(\theta)}{\epsilon},
\qquad j=1,\ldots,m .
\end{equation}
By construction, $q_j^\theta \geq 0$ for all $j$, and their sum satisfies
\[
\sum_{j=1}^m q_j^\theta
=
\frac{1}{\epsilon}
\left[
\sum_{j=1}^m p_{0j}
-
(1-\epsilon)\sum_{j=1}^m p_j(\theta)
\right]
=
\frac{1 - (1-\epsilon)}{\epsilon}
=1.
\]
Thus, $\bq^\theta$ is a valid probability vector on the support atoms, and we can write $\bp_0=(1-\epsilon)\bp(\theta)+\epsilon \bq^\theta$. If $\epsilon=0$, the inequality implies $\bp_0=\bp(\theta)$, which can be interpreted within the same form. The total variation distance between $\bp(\theta)$ and $\bp_0$ is then bounded by
\begin{equation}
\begin{aligned}
d_{\mathrm{TV}}\{\bp(\theta),\bp_0\}
&= \frac{1}{2}\sum_{j=1}^m |p_j(\theta) - p_{0j}| = \frac{1}{2}\sum_{j=1}^m
|p_j(\theta) - (1-\epsilon)p_j(\theta) - \epsilon q_j^\theta| \\
&= \epsilon \cdot \frac{1}{2}
\sum_{j=1}^m |p_j(\theta) - q_j^\theta| \leq \epsilon,
\end{aligned}
\label{eqn: proof tv intermediate}
\end{equation}
where the final inequality follows because the total variation distance between any two probability vectors is at most 1.

Now, let $\theta^*\in\arg\min_\theta\epsilon_{\min}(\theta)$. Since $\theta^*$ minimizes $\epsilon_{\min}(\theta)$, we have $\epsilon_{\min}(\theta^*) \leq \epsilon_{\min}(\theta_0) \leq \epsilon_0$. Applying \eqref{eqn: proof tv intermediate} with $\theta = \theta^*$ and $\epsilon = \epsilon_0$ yields $d_{\mathrm{TV}}\{\bp(\theta^*),\bp_0\}\leq\epsilon_0$. Furthermore, the true population representation directly implies $d_{\mathrm{TV}}\{\bp(\theta_0),\bp_0\}\leq\epsilon_0$. By the triangle inequality for total variation distance, we obtain
\begin{equation}
d_{\mathrm{TV}}\{\bp(\theta^*),\bp(\theta_0)\}
\leq
d_{\mathrm{TV}}\{\bp(\theta^*),\bp_0\} + d_{\mathrm{TV}}\{\bp(\theta_0),\bp_0\}
\leq
2\epsilon_0,
\end{equation}
proving the general total-variation stability bound.

For the specific case of $P_\theta=\operatorname{Binomial}(M,\theta)$,
\begin{equation}
|\theta^*-\theta_0|
=
\Big|
\sum_{j=0}^M \frac{j}{M}\{p_j(\theta^*)-p_j(\theta_0)\}
\Big|
\leq
d_{\mathrm{TV}}\{\bp(\theta^*),\bp(\theta_0)\},
\end{equation}
because $j/M\in[0,1]$ and total variation is the supremum of expectation differences over functions bounded between zero and one. Hence $|\theta^*-\theta_0|\leq 2\epsilon_0$, completing the proof.
\end{proof}

\begin{proof}[Proof of Theorem \ref{thm: posterior weak limit}]
By the uniform part of Theorem~\ref{thm: likelihood asymptotics},
\begin{equation}
n^{m-b-1}\bar L_b(\theta)
\to
\Gamma(m)A_{0,b}(\theta)
\end{equation}
uniformly over $\Theta$ under the stated lower bound on the atom probabilities and upper bound $\rho_0(\theta)\leq 1-\xi$.
Let $h$ be any bounded continuous test function. Multiplying the numerator and denominator of the posterior expectation by $n^{m-b-1}$ gives
\begin{equation}
\int h(\theta)\,\Pi(d\theta\mid N)
=
\frac{\int_\Theta h(\theta)n^{m-b-1}\bar L_b(\theta)\pi(\theta)\,d\theta}
{\int_\Theta n^{m-b-1}\bar L_b(\theta)\pi(\theta)\,d\theta}.
\end{equation}
Uniform convergence and continuity of $\pi$ imply convergence of both integrals, yielding
\begin{equation}
\int h(\theta)\,\Pi(d\theta\mid N)
\to
\frac{\int_\Theta h(\theta)A_{0,b}(\theta)\pi(\theta)\,d\theta}
{\int_\Theta A_{0,b}(\theta)\pi(\theta)\,d\theta}.
\end{equation}
The denominator is positive because $A_{0,b}(\theta)>0$ whenever $\rho_0(\theta)>0$ and the prior density is positive on $\Theta$.
This proves weak convergence to the distribution with density \eqref{eqn: posterior partial identification limit}.
\end{proof}

\begin{proof}[Proof of Theorem \ref{thm: sparse contamination}]
Since $(\epsilon_0,Q)$ is not unique, we take $\epsilon_0=\epsilon_{\min}(\theta_0)$ and, when $\epsilon_0>0$, $q_j=\{p_{0j}-(1-\epsilon_0)p_j(\theta_0)\}/\epsilon_0$. Then $S=\{j:q_j=0\}$ is the set of atoms attaining the minimum of $p_{0j}/p_j(\theta_0)$, so $|S|=s_0$, and $S$ is a proper subset of the atoms when $\epsilon_0>0$ since $\sum_jq_j=1$. Fix $\delta>0$. By Chebyshev's inequality there is $C_\delta$ such that $\max_j|N_j-np_{0j}|\leq C_\delta\sqrt n$ with probability at least $1-\delta$, and all bounds below hold on this event. We bound $\bar L_b(\theta_0)$ from below and $\bar L_b(\theta)$ from above by powers of $n$, using Lemma~\ref{lem: representation}. Part (a) compares the two exponents, and part (b) makes the upper bound uniform outside a neighborhood of $\theta_0$ and the lower bound uniform near $\theta_0$. Since $\alpha_j\leq\alpha_0\leq1$, each $\kappa_j$ is nonincreasing with $\kappa_j(0)=1$, so $\kappa_j\leq1$.

For the lower bound at $\theta_0$, we write $p_j=p_j(\theta_0)$, with $S$ all atoms when $\epsilon_0=0$. Let $\mathcal R_n$$=\{R:|R-n(1-\epsilon_0)|\leq\sqrt n\}$ if $\epsilon_0>0$ and $\mathcal R_n$$=\{n\}$ if $\epsilon_0=0$. Since $p_{0j}=(1-\epsilon_0)p_j$ on $S$, we have $|N_j-Rp_j|\leq C'\sqrt n$ for $R\in\mathcal R_n$ and $j\in S$. We keep only the allocations with no contamination on $S$. With $\kappa_j(0)=1$ and $\kappa_j(N_j-Z_j^{(R)})\geq \kappa_j(n)\geq cn^{\alpha_j-1}$ by Lemma~\ref{lem: representation},
\begin{equation}
E_{\theta_0}\Big[\prod_{j=1}^m \kappa_j(N_j-Z_j^{(R)})\mathbf{1}\{Z_j^{(R)}\leq N_j,\ \forall j\}\Big]
\geq
c\,n^{-\sum_{j\notin S}(1-\alpha_j)}\Pr\{Z_S^{(R)}=N_S,\ Z_j^{(R)}\leq N_j,\ \forall j\notin S\}.
\end{equation}
Since $N_S$ is within $O(\sqrt n)$ of its mean, the local limit theorem for the multinomial distribution, whose cells $j\in S$ and the remaining cell have positive probabilities, gives $\Pr(Z_S^{(R)}=N_S)\geq cn^{-s_0/2}$ when $\epsilon_0>0$ and $\Pr(Z^{(n)}=N)\geq cn^{-(m-1)/2}$ when $\epsilon_0=0$. Given $Z_S^{(R)}=N_S$, each $Z_j^{(R)}$ with $j\notin S$ is binomial with mean $Rp_j+O(\sqrt n)$, while $N_j=Rp_j+n\epsilon_0q_j+O(\sqrt n)$ with $q_j>0$, so Lemma~\ref{lem: binomial tail} shows that the remaining conditional probability tends to one. If $\epsilon_0>0$, then $W_{n,R}\geq cn^{b-\alpha_0}$ on $\mathcal R_n$, and summing over $\mathcal R_n$ gives the exponent $b+1/2-s_0/2-(m-s_0)-\sum_{j\in S}\alpha_j$, where $\sum_{j\in S}\alpha_j\leq1$. If $\epsilon_0=0$, then $W_{n,n}\geq cn^{b-1}$ and $s_0=m$. In both cases
\begin{equation}
\bar L_b(\theta_0)\geq c\,n^{\,b-1/2-s_0/2-(m-s_0)}.
\label{eqn: sparse lower bound}
\end{equation}

For the upper bound at $\theta$, we fix $\theta$ with $s=s(\theta)<m$, so that $\bp(\theta)\neq \bp_0$ and $\rho_0=\rho_0(\theta)<1$, and let $J^*$ be the set of atoms attaining the minimum. Let $\gamma_0=\min_{j\notin J^*}\{p_{0j}-\rho_0p_j(\theta)\}>0$ and choose $\eta<\min(\rho_0,1-\rho_0,\gamma_0)/2$, so that $p_{0j}-\rho p_j(\theta)\geq\gamma_0/2$ for $j\notin J^*$ and $|\rho-\rho_0|\leq\eta$, where $\rho=R/n$. As in the proof of Theorem~\ref{thm: likelihood asymptotics}, we apply Lemma~\ref{lem: binomial tail} to the multinomial marginals on three ranges of $R$, with $\sum_{R\leq n}W_{n,R}\leq Cn^{b+1-\alpha_0}$. For $\rho>\rho_0+\eta$, some $j\in J^*$ has $N_j-Rp_j(\theta)\leq-cn$. For $\rho<\rho_0-\eta$, every $N_j-Z_j^{(R)}$ is at least $cn$ outside probability $e^{-cn}$. For $|\rho-\rho_0|\leq\eta$, we have $W_{n,R}\leq Cn^{b-\alpha_0}$, the atoms outside $J^*$ have $N_j-Z_j^{(R)}\geq cn$ outside probability $e^{-cn}$, and the atoms in $J^*$ use $\kappa_j\leq1$. With $\kappa_j(k)\leq C(k+1)^{\alpha_j-1}$ from Lemma~\ref{lem: representation}, the terms of \eqref{eqn: marginal likelihood homogeneous} over the three ranges sum to at most
\begin{align*}
\textstyle\sum_{\rho>\rho_0+\eta}&\leq Cn^{b+1-\alpha_0}e^{-cn},\\
\textstyle\sum_{\rho<\rho_0-\eta}&\leq Cn^{b+1-\alpha_0}\big(n^{\alpha_0-m}+e^{-cn}\big)\leq Cn^{b+1-m},\\
\textstyle\sum_{|\rho-\rho_0|\leq\eta}&\leq C(2\eta n+1)\,n^{b-\alpha_0}\big(n^{-\sum_{j\notin J^*}(1-\alpha_j)}+e^{-cn}\big)\leq Cn^{b+1-(m-s)-\sum_{j\in J^*}\alpha_j}.
\end{align*}
The last bound dominates since $(m-s)+\sum_{j\in J^*}\alpha_j\leq m$. Hence
\begin{equation}
\bar L_b(\theta)\leq C\,n^{\,b+1-(m-s)-\sum_{j\in J^*}\alpha_j}.
\label{eqn: sparse upper bound}
\end{equation}

For part (a), combining \eqref{eqn: sparse lower bound} and \eqref{eqn: sparse upper bound}, the ratio is at least $cn^{e}$ with
\begin{equation}
e=\frac{s_0-3}{2}-s+\sum_{j\in J^*}\alpha_j.
\end{equation}
If $s_0\geq2s+3$, then $s<m$ and $(s_0-3)/2-s\geq0$, and $\sum_{j\in J^*}\alpha_j>0$ because $J^*$ is not empty, so $e>0$. The term $s_0/2$ comes from the local limit over $S$ in the lower bound, and $s$ from the atoms in $J^*$, which carry no factor $n^{\alpha_j-1}$ in the upper bound.

For part (b), we fix an open neighborhood $U$ of $\theta_0$ and let $s_1=\max_{\theta\in\Theta\setminus U}s(\theta)$, so that $s_0\geq2s_1+d+3$ by assumption. Let $K=\Theta\setminus U$. There is $\gamma>0$ such that, for every $\theta\in K$, at most $s_1$ atoms satisfy $p_{0j}/p_j(\theta)<\rho_0(\theta)+\gamma$, since otherwise a limit point of a violating sequence would have more than $s_1$ minimizing atoms by continuity. Since $s_1<m$, compactness also gives $\sup_K\rho_0<1$, and $\rho_0\geq\min_jp_{0j}>0$. The argument for \eqref{eqn: sparse upper bound} with these atoms in place of $J^*$ and a single $\eta<\min\{\gamma\min_{\theta\in K,j}p_j(\theta),\,\inf_K\rho_0,\,1-\sup_K\rho_0\}/2$ then has constants free of $\theta$. For the set $T$ of such atoms, its exponent is $b+1-m+\sum_{j\in T}(1-\alpha_j)\leq b+1-(m-s_1)-s_1\alpha_{\min}$, since $|T|\leq s_1$ and $\alpha_{\min}=\min_j\alpha_j$, so $\sup_{\theta\in K}\bar L_b(\theta)\leq Cn^{b+1-(m-s_1)-s_1\alpha_{\min}}$. For $\|\theta-\theta_0\|\leq n^{-1/2}$, the means in the lower bound move by $O(\sqrt n)$ by the Lipschitz condition, so \eqref{eqn: sparse lower bound} holds uniformly on this ball with a larger $C'$. The ball lies in $U$ for large $n$ and has prior mass at least $cn^{-d/2}$. Therefore the posterior probability of $K$ is at most
\begin{equation}
\frac{\sup_{\theta\in K}\bar L_b(\theta)}{\int_U\bar L_b(\theta)\pi(\theta)\,d\theta}
\leq
Cn^{-e_1},
\qquad
e_1=\frac{s_0-3-d}{2}-s_1+s_1\alpha_{\min},
\end{equation}
and $s_0\geq2s_1+d+3$ gives $e_1>0$. Hence $\Pi(U\mid Y_{1:n})\to1$ in probability, and since $U$ is arbitrary, the posterior is consistent at $\theta_0$.
\end{proof}

\section{Computational details}
\label{sec:supp-computation}

The marginal likelihood \eqref{eqn: dp likelihood}, its score \eqref{eqn: collapsed score},
and the coefficient normalization are evaluated by the forward dynamic programming recursion of
Section~\ref{sec: computation}. We record here the derivative recursion and the
normalization specific to gradient-based sampling. Suppose $\theta=(\theta_1,\ldots,\theta_d)$,
the atom probabilities are differentiable, and write
$\dot p_{j\ell}(\theta)=\partial p_j(\theta)/\partial\theta_\ell$ and
$D_{j\ell}(R;\theta)=\partial C_j(R;\theta)/\partial\theta_\ell$. Differentiating the
cross-atom recursion \eqref{eqn: dp recursion} in the main text gives the forward recursion
\begin{equation}
D_{j\ell}(R;\theta)
=
\sum_{r=0}^{\min\{R,N_j\}}
\left[
D_{j-1,\ell}(R-r;\theta)\frac{p_j(\theta)^r}{r!}
+
C_{j-1}(R-r;\theta)\frac{p_j(\theta)^{r-1}}{(r-1)!}\dot p_{j\ell}(\theta)
\right]\frac{(\alpha_j)_{N_j-r}}{(N_j-r)!},
\label{eqn: supp derivative dp recursion}
\end{equation}
with $D_{0\ell}(R;\theta)=0$ and the convention that the second term vanishes at $r=0$.
This recursion has the same $O\{n(n+m)\}$ cost as the likelihood recursion, and the resulting
score is \eqref{eqn: collapsed score}. The derivative arrays $D_{j\ell}(R;\theta)$ are
normalized by the same accumulated scale factors as the coefficients $C_j(R;\theta)$, so
the score ratio is unaffected by the rescaling. Equivalently, the recursion can be carried out in log space, since $R!\,C_j(R;\theta)\leq1$ when $\alpha_j\leq1$. For $\theta\in\mathbb{R}^d$, evaluating the
full score costs $d$ times the scalar derivative recursion. All beta-binomial results in the
article use a Metropolis-adjusted Langevin algorithm targeting the posterior with score
\begin{equation}
\nabla_\theta \log \pi(\theta\mid N)
=
\nabla_\theta \log \pi(\theta)
+
\nabla_\theta \log \bar L_b(\theta).
\label{eqn: supp collapsed posterior score}
\end{equation}
On a laptop, one likelihood and score evaluation takes about $5$ milliseconds for $n=500$ and $m=41$, and about $14$ milliseconds for $n=500$ and $m=425$. Since the posterior can have more than one mode, in Section~\ref{sec: real data}, pooled draws from three chains with dispersed starting values reproduce the posterior means of $\mu$ in Table~\ref{tab: bfi}.
\section{Validation of the dynamic programming algorithm}
\label{sec:supp-validation}

The dynamic programming algorithm of Section~\ref{sec: computation} is validated against brute-force enumeration of the collapsed likelihood on a small instance.

\subsection{Brute-force enumeration of the collapsed likelihood}

On small instances the marginal likelihood can be evaluated by direct
enumeration of all clean/contaminated allocations. For counts
$(N_1,\ldots,N_m)$ the sum ranges over all $\prod_j (N_j+1)$ vectors
$(r_1,\ldots,r_m)$ with $0\le r_j\le N_j$, each term combining the
structural factor $\prod_j p_j(\theta)^{r_j}/r_j!$, the per-atom Dirichlet
factors, and the Beta factor in the total contaminated count. We compared
the dynamic programming algorithm against this enumeration over $80$ configurations
crossing two structural models (binomial and beta-binomial on $m=4$ atoms),
four count vectors with $n\le 9$ including empty and concentrated atoms,
four hyperparameter pairs $(b,\alpha)\in\{(1,1),(4,\tfrac12),(1,2),(1,\tfrac14)\}$, the last with $\alpha=1/m$, and
two or three values of $\theta$ per model. The maximum absolute discrepancy
in the log-likelihood was $3.6\times 10^{-15}$, machine precision. The
analytic score \eqref{eqn: collapsed score} was compared with central finite
differences of the enumerated log-likelihood. The maximum relative error was
$2.0\times 10^{-7}$. The dynamic programming algorithm is thus validated independently of any sampler.

\section{Sensitivity to the priors}
\label{sec:supp-sensitivity}

Reported proportions of careless responding vary widely across studies \citep{WardMeade2023,ArthurEtAl2021}. A prior mean $\bar\epsilon$ for $\epsilon$ is matched by $b=1/\bar\epsilon-1$, for example $b=9$ for $10\%$ and $b=4$ for $20\%$.

Figure~\ref{fig: s1 alpha} compares $\alpha_j=1/m$ with $\alpha_j=1$ in the item-level simulation. Under SL and PATTERN, $\alpha_j=1/m$ has a smaller root mean squared error than $\alpha_j=1$, and without contamination, the two yield similar results.

Figure~\ref{fig: e5 overlap} varies the overlap in the bounded-count model of Section~\ref{sec: bounded count simulation}, with $\theta_c\in\{0.35,0.40,0.45,0.60,0.75\}$ and $\alpha_j\in\{1/m,1\}$. When the contaminating distribution is diffuse, $\alpha_j=1$ performs better, in line with Section~\ref{sec: theory}. 

Table~\ref{tab: e4 eps prior} reports coverage under $\alpha_j=1/m$ for priors on $\epsilon$ that favor smaller values than $\epsilon_0$, such as $\operatorname{Beta}(1,99)$, and higher values than $\epsilon_0$, such as $\operatorname{Beta}(5,5)$. With separated contamination, coverage stays near the nominal level up to moderate $\epsilon_0$ and decreases at the largest $\epsilon_0$, most for $\operatorname{Beta}(1,99)$, which places little mass near $\epsilon_0$. In the overlapping regime, no prior on $\epsilon$ restores coverage.

Table~\ref{tab: bfi eps prior} repeats the application of Section~\ref{sec: real data} under $b\in\{1,4,9\}$. The posterior mean of $\epsilon$ decreases slightly as $b$ grows, and that of $\mu$ barely changes.

Figure~\ref{fig: s1 eps} shows the mean estimate of $\epsilon$ in the item-level simulation of Section~\ref{sec: item simulation}. Without contamination, the posterior mean of $\epsilon$ is near zero, and under SL and PATTERN it tracks $\epsilon_0$. Across the 16 designs without contamination, the coverage of ExHuber is between $0.93$ and $0.99$ and the posterior mean of $\epsilon$ is between $0.017$ and $0.025$.

\begin{figure}[t]
\centering
\includegraphics[width=\linewidth]{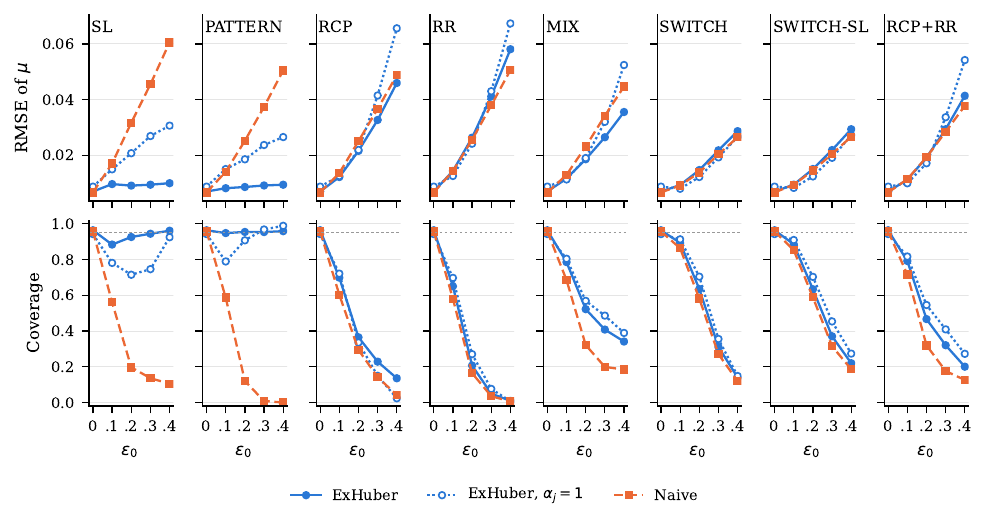}
\caption{Item-level simulation, $\alpha_j=1/m$ against $\alpha_j=1$ and the naive posterior, averaged over 16 designs. SWITCH-SL changes from attentive to straightlining after the fifth item, and RCP+RR has equal shares of RCP and RR.}
\label{fig: s1 alpha}
\end{figure}

\begin{figure}[t]
\centering
\includegraphics[width=0.72\linewidth]{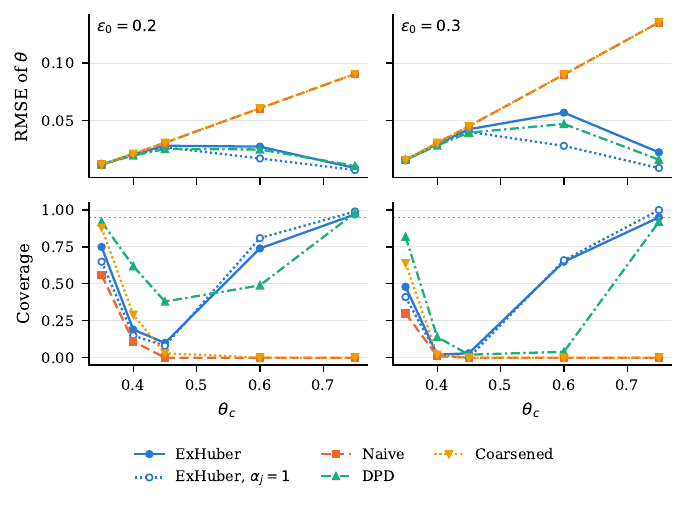}
\caption{Bounded-count model. Root mean squared error (top) and coverage (bottom) for $\theta$ against the contaminating mode $\theta_c$, for $\epsilon_0=0.2$ (left) and $\epsilon_0=0.3$ (right), with $\operatorname{Beta}(1,1)$ on $\epsilon$. The dotted line denotes $0.95$.}
\label{fig: e5 overlap}
\end{figure}

\begin{table}[t]
\centering
\caption{Coverage for $\theta$ in the bounded-count model under $\alpha_j=1/m$, by prior on $\epsilon$. Each cell averages over 100 replications.}
\label{tab: e4 eps prior}
\begin{tabular}{ccrrrrrr}
\hline
$\theta_c$ & $\epsilon_0$ & $(1,1)$ & $(1,4)$ & $(1,9)$ & $(1,19)$ & $(1,99)$ & $(5,5)$ \\
\hline
$-$ & $0$ & 0.96 & 0.96 & 0.96 & 0.95 & 0.95 & 0.97 \\
$0.75$ & $0.1$ & 0.94 & 0.94 & 0.94 & 0.94 & 0.94 & 0.96 \\
$0.75$ & $0.2$ & 0.97 & 0.97 & 0.97 & 0.96 & 0.94 & 0.97 \\
$0.75$ & $0.3$ & 0.95 & 0.95 & 0.94 & 0.93 & 0.90 & 0.96 \\
$0.75$ & $0.4$ & 0.88 & 0.89 & 0.93 & 0.94 & 0.83 & 0.90 \\
$0.45$ & $0.1$ & 0.59 & 0.59 & 0.59 & 0.59 & 0.53 & 0.67 \\
$0.45$ & $0.2$ & 0.10 & 0.10 & 0.10 & 0.08 & 0.04 & 0.25 \\
$0.45$ & $0.3$ & 0.03 & 0.02 & 0.02 & 0.02 & 0.00 & 0.08 \\
$0.45$ & $0.4$ & 0.01 & 0.00 & 0.00 & 0.00 & 0.00 & 0.02 \\
\hline
\end{tabular}
\end{table}

\begin{table}[t]
\centering
\caption{Big Five scales. Posterior means of the engaged mean $\mu$ and the contamination proportion $\epsilon$ under ExHuber with $\alpha_j=1/m$ and $\epsilon\sim\operatorname{Beta}(1,b)$.}
\label{tab: bfi eps prior}
\begin{tabular}{lrrrrrr}
\hline
 & \multicolumn{3}{c}{$\mu$} & \multicolumn{3}{c}{$\epsilon$} \\
Scale & $b=1$ & $b=4$ & $b=9$ & $b=1$ & $b=4$ & $b=9$ \\
\hline
EX & 0.498 & 0.498 & 0.498 & 0.15 & 0.14 & 0.13 \\
AG & 0.645 & 0.644 & 0.643 & 0.15 & 0.14 & 0.13 \\
CON & 0.701 & 0.699 & 0.697 & 0.28 & 0.28 & 0.26 \\
NEU & 0.475 & 0.475 & 0.475 & 0.12 & 0.12 & 0.11 \\
OP & 0.614 & 0.614 & 0.614 & 0.08 & 0.08 & 0.08 \\
\hline
\end{tabular}
\end{table}

\begin{figure}[t]
\centering
\includegraphics[width=\linewidth]{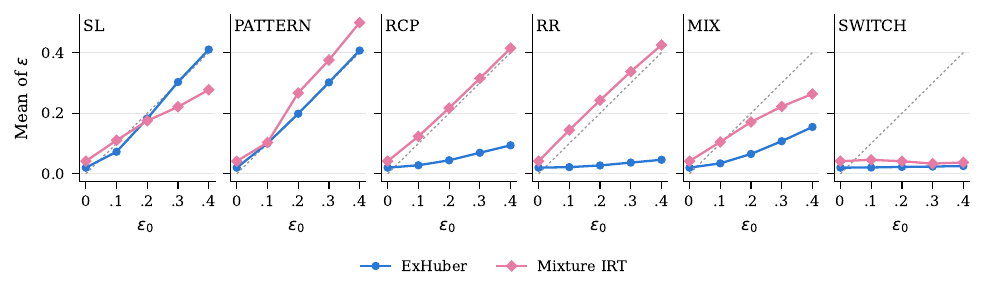}
\caption{Item-level simulation, averaged over 16 designs. Mean estimate of the contamination proportion $\epsilon$ against $\epsilon_0$, for ExHuber (posterior mean) and mixture IRT (point estimate). The dotted line denotes $\epsilon=\epsilon_0$.}
\label{fig: s1 eps}
\end{figure}

\bibliographystyle{apacite}
\bibliography{ref}

\end{document}